\documentclass[a4paper,12pt]{amsart}

\usepackage{amssymb}
\usepackage[numbers,sort&compress,merge]{natbib}
\usepackage{graphicx}

\newcommand\eg{\textit{e.g.}}

\newcommand{\abs}[1]{\left \lvert #1 \right \rvert}

\DeclareMathOperator{\var}{Var}
\DeclareMathOperator{\trace}{Tr}
\DeclareMathOperator{\sgn}{sgn}
\DeclareMathOperator{\Pf}{Pf}
\DeclareMathOperator{\lcm}{lcm}
\DeclareMathOperator*{\res}{Res}

\theoremstyle{plain}
\newtheorem{theorem}{Theorem}[section]
\newtheorem{corollary}[theorem]{Corollary}
\newtheorem{proposition}[theorem]{Proposition}

\newtheorem{lemma}[theorem]{Lemma}

\theoremstyle{definition}

\newtheorem{remark}[theorem]{Remark}

\numberwithin{equation}{section}
\begin{document}
\title{Tau-Function Theory of Chaotic Quantum Transport with
  $\beta=1,2,4$} 
    
\author{F.~Mezzadri}
\address{School of Mathematics, University of Bristol, Bristol BS8
1TW, UK}
\email{f.mezzadri@bristol.ac.uk}

\author{N.~J.~Simm}
\address{School of Mathematical Sciences, Queen Mary, University of
 London, Mile End Road, London E1 4NS, UK}
\email{n.simm@qmul.ac.uk}

 \thanks{Research partially supported by EPSRC, grant no: EP/G019843/1}

\begin{abstract}
  We study the cumulants and their generating functions of the
  probability distributions of the conductance, shot noise and Wigner
  delay time in ballistic quantum dots. Our approach is based on the
  integrable theory of certain matrix integrals and applies to all the
  symmetry classes $\beta\in \left \{1,2,4\right \}$ of Random Matrix
  Theory. We compute the weak localization corrections to the mixed
  cumulants of the conductance and shot noise for $\beta=1,4$, thus
  proving a number of conjectures of Khoruzhenko \textit{et
    al.}~\cite{KSS09}.  
  We derive differential equations that characterize the cumulant
  generating functions for all $\beta \in \left \{1,2,4\right \}$.
  Furthermore, when $\beta=2$ we show that the cumulant generating
  function of the Wigner delay time can be expressed in terms of the
  Painlev\'e III${}^{\prime}$ transcendant. This allows us to study
  properties of the cumulants of the Wigner delay time in the
  asymptotic limit $n \to \infty$. Finally, for all
  the 
  symmetry classes and for any number of open channels, we derive a
  set of recurrence relations that are very efficient for computing
  cumulants at all orders. 
\end{abstract}
\maketitle
\tableofcontents

\section{Introduction}
\label{sec:intro}
\subsection{Background}
\label{sse:bg}
Many successful approaches to tackle problems in quantum transport are
based on the the Landauer-B\"uttiker scattering theory of electronic
conduction (see, for example, the review article~\cite{Bee97} and
references therein).  It applies to two very general types of systems:
the first are mesoscopic conductors confined in space, often referred 
to as \textit{quantum dots}, through which the electric current flows
via two point contacts; the second  are quasi-one dimensional
wires containing scattering impurities.  The physical dimensions of
these systems are small enough that the quantum mechanical phase
coherence of the electron is important and its properties cannot be
understood in terms of classical mechanics; at the same time they are
large enough that a statistical description of the electrical current
becomes meaningful.

In 1985 Altshuler~\cite{Alt85} and independently Lee and
Stone~\cite{LS85} discovered that the statistical fluctuations of the
conductance in disordered wires are \textit{universal}.  More
precisely, they are independent of the dimensions of the sample or the
strength of the disorder provided that the aspect ratio
$\mathrm{length/width} \gg 1 $, the length is much longer than the
electron mean free path and much shorter than the localization length. 
Soon afterwards this phenomenon was observed
experimentally~\cite{WW86}.  Approximately at the same time, in a
different area of physics, it was discovered that as $\hbar \to 0$ the
correlations of the spectra of quantum systems with a chaotic
classical limit are the same as those of the eigenvalues of random
matrices from appropriate \textit{ensembles}~\cite{BGS84,Ber85b}. It
was then realised that Random Matrix Theory (RMT) could provide the 
mathematical framework to develop a statistical theory of quantum
transport that would account for the universality of the fluctuations 
of the electric current~\cite{IWZ90a,IWZ90b,BM94,Bee93,JPB94}.  

The main hypothesis behind the RMT approach to quantum transport is
that the average dwell time of the scattering electron is much larger
than the Ehrenfest time, so that the system-specific features of the
conductor become negligible and physical quantities can computed
through appropriate ensemble averages, which depend only on the
symmetries of the system.  In this paper we focus on quantum dots
connected to the environment via \textit{ideal leads}, or,
equivalently, through \textit{ballistic point contacts.}  Under these
assumptions, the scattering matrix is uniformly distributed in one of
Dyson's circular
ensembles~\cite{BS88,BS90,Bee93,BM94,JPB94}: 
the COE ($\beta=1$), CUE ($\beta=2$) and CSE ($\beta=4$), depending on
whether the system has time-reversal and/or spin-rotational
symmetry. In superconducting (Andreev) quantum dots additional
constraints imply that the relevant symmetry classes are
those 
discovered by Zirnbauer~\cite{Zir96} and Altland and
Zirnbauer~\cite{AZ96,AZ97}.

The subjects of this paper are the cumulants and their generating
functions of the probability distributions of  
the \textit{electrical conductance,}   the \textit{shot noise,} which is the time average of
the current fluctuations at zero temperature, and  the \textit{Wigner
  delay time,} which is a measure of the extra time an electron spends
in the scattering region. The Landauer formula expresses the
conductance as a sum of the \textit{transmission eigenvalues,}
$T_1,\dotsc,T_n$, where $n$ is the number of travelling modes (quantum
channels) available to the electron wavefunction incoming, or
outgoing, at one of the two edges of the conductor.  It follows that
both conductance and shot noise are linear statistics of
$T_1,\dotsc,T_n$.  More precisely, they are
\begin{equation}
  \label{eq:c_sn}
  G/G_0 := \sum_{j=1}^n T_j \quad \text{and}\quad  P/P_0:=\sum_{j=1}^n T_j(1-T_j),
\end{equation}
respectively.  The constant $G_0:= 2e^2/h$ is the unit of quantum
conductance and $P_0:= 2eVG_0$.  Here $V$ is the voltage difference
between the two edges of the conductor. In what follows we shall
always choose units where $G_0=P_0=1$.  The Wigner delay time is the 
average of the \textit{proper delay times,}
$\tau_1,\dotsc,\tau_n$, namely
\begin{equation}
  \label{eq:w_td}
  \tau_{\mathrm{W}} := \frac{1}{n}\sum_{j=1}^n \tau_j,
\end{equation}
where in this case $n$ is the total number of quantum channels. 
As we shall discuss in detail in~Sec.~\ref{sse:c_s_w}, RMT predicts
that the random variables $T_j$ are distributed like the spectra of
matrices in the \textit{Jacobi
  Ensembles}~\cite{BM94,Bee97,For06b,DBB10}, while the proper delay
times like the inverse of the eigenvalues of matrices in the 
\textit{Laguerre Ensembles}~\cite{BFB97}.  In the limit $n \to \infty$
the conductance, shot noise and Wigner delay time converge to normal
random variables (see, \textit{e.g.},~\cite{Pol89,VMB08,VMB10}).

In this article we compute the first corrections to the semiclassical
limit $n\to \infty$, often referred to as \textit{weak localisation
  corrections}, of the mixed cumulants of $G$ and $P$ for
$\beta=1,4$.  We also calculate the leading order contributions to the
cumulants of the Wigner delay time.  In addition, we derive a set of
differential equations that are satisfied by the cumulant generating
functions for all $\beta \in \left \{1,2,4\right \}$.  In particular,
we prove that when $\beta = 2$ the cumulant generating function of
$\tau_{\mathrm{W}}$ can be expressed in terms of a solution of the 
Painlev\'e III${}^{\prime}$ equation.

For over twenty years the universal properties of both linear
and 
non-linear statistics of the transmission eigenvalues and the
proper 
delay times have been extensively studied using both
RMT~\cite{LSSS95,FS97,FSS97,SFS01,KSS09,Nov07,Nov08,
  OK08,SSS01,OK09,KP10b,KP10a,SS03,SS06,SWS07,SSW08,LV11,MS11,
  MS12,VMB08,VMB10} and periodic orbit
theory~\cite{VOdAL98,RS02,VL01,LV04,MHBHA04,MHBHA05,BHMH06,HMBH06,
  KS07,KS08,BHN08,BK10,BK11,BK12,KEBPWR11}.  Nevertheless, the higher 
moments
and 
cumulants of $G$, $P$ and $\tau_{\mathrm{W}}$ appeared to be elusive
for a long time and considerable effort was spent to investigate the
lower cumulants~\cite{FS97,SFS01,SS06,SSW08,SSS01,SS03,SWS07}.  This
is hardly surprising, as the global fluctuations of the spectra of
random matrices are particularly difficult to probe and few
results are available~\cite{Joh98,DE06,KS10,BG13a,BG13b,DP12}. In the case of the
Wigner delay time, the problem is further complicated by the fact that 
the right-hand side of Eq.~\eqref{eq:w_td} is a \textit{singular}
linear statistic on the eigenvalues of matrices in the Laguerre
Ensembles. Furthermore, the appropriate Laguerre Ensemble contains an unusual $n$-dependent exponent (see Eq.~\eqref{eq:par_lag}).
As is often the case in RMT, ensembles with symmetries $\beta=1$ and $\beta=4$ are more
difficult to deal with than those 
with $\beta=2$.

  Recently Osipov and Kanzieper~\cite{OK08,OK09} made substantial
  progress by using the theory of integrable systems combined with RMT
  to study the cumulant generating functions of the conductance and shot 
  noise.  They proved that when time-reversal symmetry is broken, the cumulant 
  generating function of the conductance can be expressed in terms of
  a solution to the Painlev\'e V equation.  They also computed the 
  leading order terms of the cumulants in the limit as $n \to \infty$.
  Independently, Novaes~\cite{Nov08} using Selberg's integral
  discovered a finite $n$ formula for the moments of the conductance
  for $\beta=2$ as a sum over partitions. Soon afterward 
  Khoruzhenko \textit{et al}~\cite{KSS09} made a further considerable
  contribution by combining generalizations of Selberg's integral with the theory of
  symmetric functions.  In particular, these techniques allowed them
  to compute the first asymptotic correction of the cumulants of
  conductance and shot noise when $\beta=2$. They were able to formulate 
  conjectures for the same quantities for ensembles with $\beta=1$
  symmetries, which until then had been much more difficult to tackle 
  than the case $\beta=2$. 
  In a recent letter Vidal and Kanzieper~\cite{VK12} computed the
  joint probability density function of the reflection eigenvalues in
  quantum dots with broken time reversal and coupled to the
  environment via point contacts with tunnel barriers.

  Our approach is based on the integrable theory of
  certain 
  classes of matrix integrals first introduced
  in~\cite{MM90,GMMMO91,TSY96} and then developed to its full
  potential by Adler, Shiota and
  van 
  Moerbeke~\cite{ASvM95,ASvM98,ASvM02} and by Adler and van
  Moerbeke~\cite{AvM95,AvM01a,AvM01b,AvM02}.  Using this formalism, we
  show that the mixed cumulant generating function of the
  conductance 
  and shot noise satisfies a certain non-linear PDE, which can be
  reduced to a set of recurrence relations for the cumulants of the
  conductance and the joint cumulants of conductance and shot
  noise. By looking at the leading order contribution to such
  difference equations, we can prove the conjectures of
  Khoruzhenko 
  \textit{et al.}~\cite{KSS09} on the cumulants of conductance and shot
  noise when $\beta=1$, as well as perform the same analysis when
  $\beta=4$.

  These techniques are powerful enough to cope with singularities of
  the linear statistics~\eqref{eq:w_td} too.  Indeed, we can apply the
  same approach to the distribution of $\tau_{\mathrm{W}}$ and prove a
  number of finite $n$ properties that were previously beyond reach.
  These include recurrence relations for the cumulants for all $\beta
  \in \{1,2,4\}$, as well as the link between the Painlev\'e  $\mathrm{III}^{\prime}$ 
  equation and the cumulant generating function when
  $\beta=2$. 

\subsection{The Landauer-B\"uttiker Theory and Random Matrix Theory} 
\label{sse:c_s_w} 

The dynamics of an electron in a mesoscopic conductor can be studied
by looking at its scattering by a two-dimensional ballistic cavity
whose classical dynamics is chaotic.  The scattering region is
connected to two electron reservoirs in equilibrium at zero
temperature and Fermi energy $E_{\mathrm{F}}$ by two ideal leads,
which support a finite number of propagating modes (quantum channels).
At low voltage the electron-electron interactions are negligible and
the scattering is elastic. Therefore, the scattering matrix at energy
$E$ relative to $E_{\mathrm{F}}$ is unitary and has the block
structure
\begin{equation}
\label{scatterintro}
S(E) := 
\begin{pmatrix} 
r_{m \times m} & t_{m \times  n}'\\
t_{n \times m} & r_{n \times n}' 
\end{pmatrix}.
\end{equation}
Here $m$ and $n$ are the numbers of channels in the left and right
leads, respectively; the sub-blocks $r_{m\times m}$, $t_{n\times m}$ are
the reflection and transmission matrices through the left lead and
$r'_{n\times n}$ and $t'_{m \times n}$ those through the right
lead.  Without loss of generality we shall assume that $m \ge n$.  The
unitary matrix in~\eqref{scatterintro} is often referred to as 
the Landauer-B\"uttiker scattering matrix.  In this setting the
semiclassical limit $\hbar \to 0$ is equivalent to $n \to \infty$
while the ratio $m/n$ remains finite.

The transmission eigenvalues $T_1,\dotsc,T_n$ are the eigenvalues of
the Hermitian matrix $tt^\dagger$ as well as of $t^{\prime
  \dagger}t'$; thus, formulae~\eqref{eq:c_sn} become
\begin{equation}
  \label{eq:sec_con}
  G = \trace tt^\dagger \quad \text{and} \quad 
 P=\trace \bigl(tt^\dagger(I - tt^{\dagger})\bigr). 
\end{equation}
Since $S(E)$ is unitary, $T_j \in [0,1]$ for $j=1,\dotsc,n$. The proper 
delay times are the eigenvalues of the Wigner-Smith time-delay matrix,
which is defined by
\begin{equation}
  \label{eq:wdtm}
  Q := -i\hbar S^{-1}(E)\frac{\partial S(E)}{\partial E}.
\end{equation}
Thus,  the Wigner delay time is the normalized trace
\begin{equation}
  \label{eq:w_dt2}
  \tau_{\mathrm{W}} = \frac{1}{n}  \trace Q.
\end{equation}
Note that when we discuss the conductance and shot noise, $n$ denotes 
the number of quantum channels in one lead and not the dimensions of
the scattering matrix, which are $(m + n) \times (m+n)$. When we 
consider the Wigner delay time, $n$ is the total number of channels;
thus, in this case the dimensions of $S(E)$ and of $Q$ are $n \times
n$.  We adopt this convention to simplify the notation in the rest of the
paper, since $n$ will always be the asymptotic parameter.

Within the RMT approach to quantum transport, the transmission matrix 
$tt^\dagger$ belongs to one of the Jacobi Ensembles; therefore, the 
joint probability density function (\textit{j.p.d.f.}) of
$T_1,\dotsc,T_n$ is~\cite{BM94,Bee97,For06b,DBB10}
 \begin{equation}
  \label{eq:tr_eig_Aqd}
  P_{\mathrm J}(T_1,\dotsc,T_n) = \frac{1}{c^{(\beta)}_n}
\prod_{j=1}^{n}T_{j}^{\alpha}\left(1 - T_j\right)^{\delta/2} 
  \prod_{1 \leq j < k \leq n}\left.
    \lvert T_{k}-T_{j}\right \rvert^{\beta},
\end{equation}
where $\alpha =\frac{\beta}{2}\left(m-n+1\right)-1$ and
\begin{equation}
\label{eq:selberg}
\begin{split}
   c_n^{(\beta)} & := \int_{[0,1]^n}\prod_{j=1}^{n}T_{j}^{\alpha}
    \left(1 - T_j\right)^{\delta/2} 
  \prod_{1 \leq j < k \leq n}\left.
    \lvert T_{k}-T_{j}\right \rvert^{\beta}dT_1\dotsm dT_n \\
   & = \prod_{j=0}^{n-1}\frac{\Gamma\left(\alpha + 1
        +j\beta/2\right)\Gamma\left(\delta/2 + 1 + j\beta/2\right)
        \Gamma\bigl(1 + (j+1)\beta/2\bigr)}
         {\Gamma\bigl(\alpha + \delta/2 + 2 +\left(n
             +j-1\right)\beta/2\bigr) 
          \Gamma\left(1 + \beta/2\right)}
\end{split}
\end{equation}
is Selberg's integral (see, \textit{e.g.},~\cite{For10}, Ch. 4). 
If the scattering matrix $S(E)$ belongs to one of Dyson's
ensembles then $\delta=0$. When additional constraints are imposed new 
symmetry classes may arise.  In mesoscopic physics this happens when a
chaotic cavity is in contact with a superconductor (Andreev quantum
dots). The symmetry classes for these systems were predicted by
Zirnbauer~\cite{Zir96} and Altland and Zirnbauer~\cite{AZ96,AZ97}.
Due\~{n}ez~\cite{Due04} studied these symmetry classes too, and
extended Zirnbauer's theory.  The new ensembles modelling the
scattering matrix of Andreev quantum dots are the compact symmetric
spaces that, according to Cartan's classification, are of symmetry
types C, CI, D and DIII.  Such symmetries affect the dependence on the
integers $(\beta, \delta)$ of the
\textit{j.p.d.f.}~\eqref{eq:tr_eig_Aqd}.  More precisely, the
connection between the symmetry class of $S(E)$ and $(\beta,\delta)$
is the following~\cite{DBB10}:
\begin{align*}
  (\beta,\delta) & = (4,2), && \text{symmetry type C;} &
    (\beta,\delta) & = (2,1), && \text{symmetry type CI;}\\
  (\beta,\delta) & = (1,-1), && \text{symmetry type D;} &
  (\beta,\delta) & = (2,-1), && \text{symmetry type DIII.}
\end{align*}
Physically, they correspond to different
combinations of spin-rotation and time-reversal symmetries.\footnote{It is
  important to emphasise that while $\beta=1$ and $\beta=4$ label
  Wigner-Dyson symmetry classes that are time-reversal invariant, in
  the Altland-Zirnbauer classification~\cite{Zir96,AZ96,AZ97,Due04}
  they refer to systems in which time-reversal symmetry is
  broken~\cite{AZ97,DBB10}.}

If the $T_1,\dotsc,T_n$ were independent random variables, then the
central limit theorem asserts that the variances of $G$ and $P$ would
grow in proportion to $n$; instead, it remains finite in the limit
$n 
\to \infty$.  More specifically, we
have~\cite{IWZ90a,Bee93,BM94,JPB94}
\begin{equation}
\label{gpvar}
\lim_{n \to \infty}\var(G) = \frac{1}{8\beta} \quad \text{and} \quad
 \lim_{n \to \infty}\var(P) = \frac{1}{64\beta}.
\end{equation} 
This phenomenon is a manifestation of the \textit{universal
  conductance fluctuations} and is due to the strong correlations
among the transmission eigenvalues, caused by the factor $\prod_{1\le j 
  < k \le n}\abs{T_k - T_j}^\beta$ in the right-hand side
of~\eqref{eq:tr_eig_Aqd}.  It is a common property for the spectra
of 
random matrices.  Similar central limit theorems on linear statistics
of eigenvalues have been proved for the \textit{classical compact
  groups}~\cite{DS94,DE01,Joh97}, the Gaussian, Jacobi and 
Laguerre Ensembles~\cite{Joh98,DE06,DP12}.

Braun \textit{et al.}~\cite{BHMH06} and Heusler \textit{et
  al.}~\cite{HMBH06} developed a semiclassical derivation of the
average and variance of $G$ as well as of the average of $P$, which at
the time was still unknown, to all orders in $1/n$.  Shortly after, 
Savin and Sommers~\cite{SS06} reproduced the semiclassical results
using Selberg's integral.  Subsequently, the same techniques allowed
Sommers \textit{et al.}~\cite{SWS07} and Savin \textit{et al.}~\cite{SSW08} to compute
the variance of $P$ and the first four cumulants of $G$
non-perturbatively.

The probability distributions of $P$ and $G$ are strongly non-Gaussian
for small numbers of quantum channels \cite{SWS07,KSS09,KP10b}. In the 
large-$n$ limit they approach a universal Gaussian curve and contain
weak singularities which have been investigated using large deviation
estimates by Vivo \textit{et al.}~\cite{VMB08,VMB10}. Furthermore,
Khoruzhenko \textit{et al.}~\cite{KSS09} derived exact Fourier series
representations of these distributions; however, the coefficients of
such series involve $n \times n$ determinants ($\beta=2$) or
Pfaffians ($\beta=1$ or $\beta=4$) which are difficult to
handle explicitly.

When the scattering matrix belongs to one of Dyson's circular
ensembles, then the \textit{j.p.d.f.} of the proper delay
times is~\cite{BFB97}  
\begin{equation}
\label{eq:wdt_d}
P_{\mathrm W}(\tau_{1},\ldots,\tau_{n}) = \frac{1}{Z}
\prod_{j=1}^{n}\tau_{j}^{-b}e^{-\frac{\beta \tau_{\mathrm{H}}}{2\tau_{j}}}
\prod_{1 \leq j < k \leq n}\abs{\tau_{k}-\tau_{j}}^{\beta},
\end{equation}
where $Z$ is a normalization constant, $\tau_{\mathrm H}$ is the
Heisenberg time and
\begin{equation}
\label{eq:par_lag}
b=\frac{3\beta n}{2} + 2 -\beta.
\end{equation}
In our choice of
units $\tau_{\mathrm H}=n$.  The substitution $x_j= 1/\tau_j$,
$j=1,\dotsc,n$, turns~\eqref{eq:wdt_d} into the \textit{j.p.d.f.} of
the eigenvalues of matrices in the Laguerre Ensembles.

The singularities in the right-hand side of Eq.~\eqref{eq:wdt_d} make
the Wigner delay time considerably more challenging to study than the
conductance and shot noise; indeed, much less is known about the
probability distribution of $\tau_{\mathrm{W}}$. Lehmann 
\textit{et al.}~\cite{LSSS95} studied the parametric
correlations of $\tau_{\mathrm W}$ when $\beta=1$  for
small and large number of quantum channels using supersymmetry
techniques in  RMT. Fyodorov and
Sommers~\cite{FS97} computed the finite $n$ parametric correlations as 
well as the variance when $\beta=2$. Furthermore, Fyodorov \textit{et
  al.}~\cite{FSS97} derived the crossover of the same quantities to
systems with $\beta=1$. Kuipers and Sieber~\cite{KS08} computed such
correlations as 
well as the first seven terms in the asymptotic expansion of the variance using periodic orbit 
theory for $\beta=1$ and $\beta=2$ symmetries.  Semiclassical
calculations to all orders in $1/n$
are 
still an open problem in this case. The distribution of the Wigner
delay time for small numbers of channels
as 
well as its transition to non-ideal coupling was treated
in~\cite{SFS01,SSS01,SS03} for each of the Wigner-Dyson symmetry
classes. 
Recently, Texier and Majumdar~\cite{TM13} studied the large deviations in the tails of the distribution of $\tau_{\mathrm{W}}$. 
\newpage

\section{Main Results and Discussions}
\label{se:m_r}
  \subsection{Conductance and Shot Noise}
\label{ss:c_sn_mr}
  The definitions~\eqref{eq:c_sn} and the
  \textit{j.p.d.f.}~\eqref{eq:tr_eig_Aqd} imply that the random
  variables $G$ and $P$ are correlated.  The \textit{joint moment
    generating function}
\begin{equation}
\label{intrepjointgen}
\begin{split}
\mathcal{M}^{(\beta)}_{n}(z,w) & := \frac{1}{c^{(\beta)}_{n}}
\int_{[0,1]^{n}}\prod_{j=1}^{n}T_{j}^{\alpha}(1-T_{j})^{\delta/2}\\
& \quad \times \exp\bigl(-zT_{j}-wT_{j}(1-T_{j})\bigr)
    \prod_{1\le j <k \le n}\abs{T_k - T_j}^{\beta}dT_1 \dotsb dT_n
\end{split}
\end{equation}
contains all the information on the joint probability density function
of $G$ and $P$. By definition the joint cumulants
$\kappa^{(\beta)}_{l,k}$ are the coefficients in the Taylor 
expansion 
\begin{equation}
\label{mixcumugen}
\sigma^{(\beta)}_n(z,w) := \log \mathcal{M}^{(\beta)}_{n}(z,w) = \sum_{l=0}^{\infty}
\sum_{k=0}^{\infty}\frac{(-1)^{l+k}\kappa^{(\beta)}_{l,k}}{l!k!}z^{l}w^{k}.
\end{equation}
Since $\mathcal{M}_n^{(\beta)}(0,0)=1$, $\kappa^{(\beta)}_{0,0}=0$. The
cumulants of the conductance or of the shot noise are obtained as 
special cases by setting $w=0$ or $z=0$ separately in
\eqref{mixcumugen}; throughout this paper we will reserve the notation
$\kappa^{(\beta)}_{l} := \kappa^{(\beta)}_{l,0}$ for the cumulants of the conductance. 

As $n,m \to \infty$ and $m-n$ remains finite, the averages of the conductance 
and shot noise are given by, 
\begin{equation}
  \label{eq:averages}
  \kappa_1^{(\beta)} = \frac{n}{2} + O(1) \quad \text{and} \quad 
  \kappa_{0,1}^{(\beta)} = \frac{n}{8}+O(1);
\end{equation}
the limit of the variances as $n \to \infty$ are given in
Eq.~\eqref{gpvar}.

 One of the main problems addressed in this
work concerns the asymptotics of the higher cumulants.
 
\begin{theorem}
\label{maintheorem} 
Let $\beta =1,4$ and suppose that $\alpha$ is
independent of $n$, i.e. $m=n +C$ for some constant $C$.
Then, we have
\begin{subequations}
\label{eq:mixed}
\begin{align}
\label{oddmixed}
\lim_{n \to \infty}n^{l+k}\kappa^{(\beta)}_{l,k} &=
\bigl(\delta/2-\alpha\bigr)\bigl(\alpha + \delta/2+2-\beta\bigr) 
\frac{(k+l-1)!}%
{\beta(2\beta)^{l}(4\beta)^{k}}, \\
\intertext{for $l$ odd and $(l,k) \neq (1,0)$,}
\label{evenmixed} 
  \lim_{n \to \infty}n^{l+k-1}\kappa^{(\beta)}_{l,k}  &=(-1)^k
  \left(\frac{\beta}{2}-1\right) \frac{(l+k-2)!}%
  {(4\beta)^{k}(4\beta)^{l}}\sum_{j=0}^{k}\binom{2j+l}{j+l/2}
  \binom{k}{j}(-2)^{-j}, 
\end{align}
for $l$ even and $(l,k) \notin \{(0,0), (0,1),(0,2),(2,0)\}$.  
\end{subequations}
\end{theorem} 

There are some interesting consequences of this theorem that are worth
emphasizing. 
\begin{remark}
\label{staircase}  
Let $\beta=1$ and $\delta=0$. In other words, consider normal
  (non-superconducting) quantum dots with time-reversal and
  spin-rotation symmetries. Assume also that the number of quantum
  channels is the same in both leads, \textit{i.e.} $m=n$ and
  $\alpha=-1/2$, see Eq.~\eqref{eq:tr_eig_Aqd}.  Khoruzhenko,
  Savin and Sommers~\cite{KSS09} conjectured the weak localization
  corrections of the higher cumulants of the conductance:
 \begin{subequations}
\label{eq:kss_c1}
\begin{alignat}{2}
\label{goddconj}
\lim_{n \to \infty}n^{l}\kappa^{(1)}_{l} & = \frac{(l-1)!}{2^{l+2}},
&\qquad & \text{for $l$ odd,} \\
\label{gevenconj} 
\lim_{n \to \infty}n^{l-1}\kappa^{(1)}_{l}& = 
-\frac{(l-2)!}{2^{2l +1}}\binom{l}{l/2},  && \text{for $l$ even} 
\end{alignat}
\end{subequations}
with $l >2$. Equations~\eqref{goddconj} and~\eqref{gevenconj} are
particular cases of Theorem~\ref{maintheorem}.

In the same article it was also conjectured that at leading order the
higher cumulants of the shot noise power are
 \begin{subequations}
\label{eq:kss_sn1}
\begin{alignat}{2}
\label{poddconj}
\lim_{n \to \infty}n^{k}\kappa^{(1)}_{0,k} & = \frac{(k-1)!}{2^{2k+3}},
&\qquad & \text{for $k$ odd,} \\
\label{pevenconj} 
\lim_{n \to \infty}n^{k-1}\kappa^{(1)}_{0,k}& = 
-\frac{(k-2)!}{2^{3k+1}}\binom{k}{k/2},  && \text{for $k$ even,} 
\end{alignat}
\end{subequations}
with $k>2$. When $l=0$ the sum in~\eqref{evenmixed} is given by
\begin{equation}
  \label{eq:sum}
 \begin{split}
  \sum_{j=0}^k\binom{2j}{j}\binom{k}{j}(-2)^{-j} & = 
P_k^{(0,-k -\frac12)}(-3)=
\frac{(-1)^{k/2}\sqrt{\pi}}{\Gamma\left(- \tfrac{k}{2} +
    \tfrac{1}{2}\right)\Gamma\left(\tfrac{k}{2} + 1\right)} \\
& = \begin{cases} 0, & \text{for $k$ odd,}\\
           \frac{1}{2^k}\binom{k}{k/2}, 
           & \text{for $k$ even,} 
\end{cases}
\end{split}
\end{equation}
where $P_k^{(a,b)}(x)$ denotes the Jacobi polynomial (see 
Appendix~\ref{ap:p_th} for the definition). The identity
~\eqref{eq:sum} reduces formula~\eqref{evenmixed} to the right-hand
side of~\eqref{pevenconj} when $k$ is even; however, mutual cancellations occur in the
sum~\eqref{eq:sum} when $k$ is odd, which yields
\begin{equation}
  \label{sub_lead}
\lim_{n\to \infty} n^{k-1} \kappa_{0,k}^{(1)} = 0.
\end{equation}
This limit implies that the first corrections to the odd cumulants of 
$P$ are of subleading order in $1/n$ and is consistent with the
conjecture~\eqref{poddconj}. 
These cancellations explain the unusual
``staircase'' behaviour in the decay of the cumulants of the shot
noise, first observed in~\cite{KSS09}. %

As we shall discuss in Secs.~\ref{ss:j_sn_c}, \ref{sse:l_s_n}
and~\ref{higherorder}, the last step in proving
Theorem~\ref{maintheorem} consists in solving the asymptotic limit of 
certain difference equations.  As it turns out, proving
formula~\eqref{poddconj} requires computing the next to leading order
of the mixed cumulants $\kappa^{(1)}_{l,k}$.  Although in principle
our techniques could be carried out at the next order in $1/n$, the 
calculations become substantially harder (see Sec.~\ref{higherorder}).
\end{remark}

\begin{remark}
  Theorem~\ref{maintheorem} reveals some general qualitative
  differences among the probability distributions of $G$ and $P$ for
  the various symmetry classes.
 
  In the special case $(\beta,\delta)=(2,0)$, \textit{i.e.} when the
  Landauer-B\"uttiker scattering matrix belongs to the CUE, the cumulants of conductance and shot noise were computed to leading order by Osipov and Kanzieper~\cite{OK08,OK09}. They are of subleading order in $1/n$ compared to formulae~\eqref{eq:mixed},
  which means that the convergence of $G$ and $P$ to a normal random variable is slower when $\beta=1$ and $\beta=4$. Similarly, the fact
  that the mixed cumulants are of lower order for
  $\beta=2$~\cite{OK09} implies that the correlations
  between $G$ and $P$ are stronger when $\beta=1$ and $\beta=4$.
 

  For symmetry classes with $\beta=1$, the right-hand side of
  Eq.~\eqref{evenmixed} is negative, while it is positive when 
  $\beta=4$.  Thus, in general, in the former case the probability
  distributions have a lower and wider peak around the mean than
  in the latter.
\end{remark}
Some of the intermediate results leading to Theorem~\ref{maintheorem}
are interesting in their own right and it is worth discussing the main
ideas behind the proof.  By suitably deforming the integrand
of~\eqref{intrepjointgen} with an infinite sequence of ``time variables''
$\mathbf{t}=(t_{1},t_{2},\ldots)$, it can be shown that when $\beta
=1$ or $\beta=4$, the resulting integrals are connected to an infinite
hierarchy of non-linear PDEs known as the \textit{Pfaff-KP hierarchy}
\cite{AvM01a}. The first non-trivial member of this hierarchy is the
\textit{Pfaff-KP equation}. In addition, the deformations of the 
integrals~\eqref{intrepjointgen} satisfy an infinite set of linear
PDEs, known as \textit{Virasoro constraints}. We outline this theory
in Sec. \ref{se:int_hier}. 

The Pfaff-KP equation and the Virasoro
constraints can be used to find a partial differential equation that
the cumulant generating function~\eqref{mixcumugen} satisfies.
\begin{theorem}
\label{pde:mix_csn}
Let  $\beta \in \{1,2,4\}$ and 
\begin{equation}
\label{bnbeta}
b_{n}^{(\beta)} := \begin{cases}
\frac{n(n-1)}{(n+1)(n+2)}\frac{c^{(1)}_{n-2}c^{(1)}_{n+2}}%
   {\left(c^{(1)}_{n}\right)^{2}} &\text{if $\beta=1$,} \\
0 &\text{if $\beta=2$,}\\
\frac{n}{n+1}\frac{c^{(4)}_{n-1}c^{(4)}_{n+1}}%
  {\left(c^{(4)}_{n}\right)^{2}} &\text{if $\beta=4$,}
\end{cases}
\end{equation}
where $c_n^{(\beta)}$ is Selberg's
integral~\eqref{eq:selberg}. Furthermore, define the parameters
\begin{equation}
\label{def_par}
\eta^{(\beta)}_{1,4} :=  \begin{cases}
4 & \text{if $\beta=4$,} \\
1 & \text{if $\beta\in\{1,2\},$}\\
\end{cases} \qquad \text{and} \qquad i(\beta) = \begin{cases}
1 & \mbox{if } \beta=4,\\
2 & \mbox{if } \beta=1.\\
\end{cases} 
\end{equation}
The mixed cumulant generating function~\eqref{mixcumugen} satisfies
the  nonlinear PDE
\begin{multline}
\label{diffeqshot}
\eta_{1,4}^{(\beta)}w\frac{\partial^{4}\sigma^{(\beta)}_{n}}%
{\partial z^{4}}
+6\eta_{1,4}^{(\beta)}w\left(\frac{\partial^{2}\sigma^{(\beta)}_{n}}%
{\partial z^{2}}\right)^{2}+2z\frac{\partial^{2} \sigma^{(\beta)}_{n}}%
{\partial z \partial w}\\ 
+3w\frac{\partial^{2} \sigma^{(\beta)}_{n}}{\partial w^{2}}+2\frac{\partial
  \sigma^{(\beta)}_{n}}{\partial w}
+\bigl(2\left(\alpha+\delta/2
+n\beta+2-\beta\right)-w\bigr)\frac{\partial^{2}\sigma^{(\beta)}_{n}}%
{\partial z^{2}}\\
=(12b_{n}^{(\beta)}/\beta)w\exp\left(\sigma^{(\beta)}_{n-i(\beta)}+ 
\sigma^{(\beta)}_{n+i(\beta)}-2\sigma^{(\beta)}_{n}\right).
\end{multline}
\end{theorem}
We report the explicit formulae for the parameters~\eqref{bnbeta} in
Appendix~\ref{bnapp}.

The same approach that leads to~\eqref{diffeqshot} allows us to derive 
an ODE for the conductance.
\begin{theorem}
\label{ode:con}
  Denote $\sigma_{n}^{(\beta)}(z) := \log \mathcal{M}_n^{(\beta)}(z,0)$ and let
  $b_n^{(\beta)}$, $\eta_{1,4}^{(\beta)}$ and $i(\beta)$ be the same
  coefficients introduced in~\eqref{bnbeta}
  and~\eqref{def_par}. Furthermore, define the polynomials
\begin{subequations}
\label{polysdiffeqcond}
\begin{align}
p_{0}^{(\beta)}(z) &= -\frac{n}{2}(\beta n+2\alpha+2-\beta)(\beta
n+\delta/2 +\alpha-z),\\
p_{1}^{(\beta)}(z) &= (\beta n +\alpha-\delta/2)z-(\beta
n+\alpha+\delta/2)
   (\beta n+\alpha+\delta/2+2-\beta),\\
p_{2}^{(\beta)}(z) &= -z^{3}+2(\beta n+\alpha-\delta/2)z^{2}\\
& \quad -(\beta n+\alpha+\delta/2+6-3\beta)(\beta
n+\alpha+\delta/2+2-\beta)z +\beta z.\notag
\end{align}
\end{subequations}
 Then 
\begin{multline}
\label{diffeqcond}
z^{3}\eta^{(\beta)}_{1,4}\sigma^{(\beta)\prime\prime\prime\prime}_{n}
+6\eta^{(\beta)}_{1,4}z^{3}
\bigl(\sigma^{(\beta)\prime\prime}_{n}\bigr)^{2}+\beta\Bigl(2z^{2}
\sigma^{(\beta)\prime\prime\prime}_{n}
+z\bigl(\sigma^{(\beta)\prime}_{n}\bigr)^{2}
+4z^{2}\sigma^{(\beta)\prime}_{n} \sigma^{(\beta)\prime\prime}_{n}\Bigr)\\
+\left(p^{(\beta)}_{2}(z)\sigma^{(\beta)\prime\prime}_{n}+p^{(\beta)}_{1}(z)
\sigma^{(\beta)\prime}_{n} +p^{(\beta)}_{0}(z)\right)\\
=(12b_{n}^{(\beta)}/\beta)z^{3}
\exp\left(\sigma^{(\beta)}_{n-i(\beta)}
+\sigma^{(\beta)}_{n+i(\beta)}-2\sigma^{(\beta)}_{n}\right).
\end{multline}
\end{theorem}
\begin{remark}
When $(\beta,\delta)=(2,0)$ the substitution 
\begin{equation}
\sigma^{(\beta)}_n(z) \mapsto n(n + \alpha)  + z \sigma^{(\beta) \prime}_n(z)
\end{equation}
reduces Eq.~\eqref{diffeqcond} to Painlev\'e V~\cite{OK08,OK09}.
\end{remark}

\begin{remark}
  For $\beta=1,4$, Eqs.~\eqref{diffeqshot} and~\eqref{diffeqcond} are \textit{differential-difference
    equations}, as $\sigma^{(\beta)}_{n}(z)$ appears in the right-hand side with different values of $n$. This makes the analysis in the cases $\beta=1,4$ more difficult. Note that this aspect of the equations disappears when $\beta=2$, as then $b^{(\beta)}_{n}=0$. In particular, Eq.~\eqref{diffeqshot} reduces to the zero-temperature limit of a result derived in \cite{OK09}.
    
  
\end{remark}
\begin{remark}
  Since we restrict $\beta$ to take only the values $1$, $2$ and $4$,
  Selberg's integral and the coefficient $b_n^{(\beta)}$ reduce
  to a rational function of $\alpha$, $\delta$ and $n$ (see
  Appendix~\ref{bnapp}). 
\end{remark}

By introducing the power series expansions of $\sigma_{n}^{(\beta)}(z,w)$ and of $\sigma_{n}^{(\beta)}(z)$ 
into the PDE~\eqref{diffeqshot} and the ODE~\eqref{diffeqcond}, we
obtain two difference equations for the cumulants
$\kappa_{l,k}^{(\beta)}$ and $\kappa_l^{(\beta)}$.  The
$\kappa^{(\beta)}_l$'s serve as initial conditions of the recurrence
relation for the joint cumulants $\kappa^{(\beta)}_{l,k}$. In order to prove 
Theorem~\ref{maintheorem} such recurrence relations need to be solved
in the asymptotic limit $n \to \infty$.  This will be
achieved in Secs.~\ref{se:csnfn} and~\ref{as_an}.

Why is it necessary to look at the mixed cumulant generating function
and not simply at those of the conductance and shot noise separately? The main reason is that the shot noise is quadratic in the transmission eigenvalues, and consequently, the determination of a non-linear ODE satisfied by $\sigma^{(\beta)}_{n}(0,w)$ is much more complicated than for $\sigma^{(\beta)}_{n}(z,0)$. In a certain special case, we point out that $\sigma^{(\beta)}_{n}(0,w)$ satisfies an ODE related to Painlev\'e V; this is discussed in Sec. \ref{shotproperties}. More generally, the derivation of a non-linear ODE satisfied by $\sigma^{(\beta)}_{n}(0,w)$ is currently an open problem.


\subsection{The Wigner Delay Time}
As we discussed in the Introduction, the statistical theory of the Wigner delay time presents additional difficulties compared with the conductance. Nevertheless, the connection with integrable systems is powerful enough for us to obtain substantial results.

Let us set $\lambda_j = 2\tau_j/(\beta n)$, $j=1,\dotsc,n$, and define
\begin{equation}
\label{wignermgf}
M^{(\beta)}_{n}(z) := \frac{1}{m^{(\beta)}_{n}}
\int_{\mathbb{R}_{+}^{n}}\prod_{j=1}^{n}\lambda_j^{-b}
\exp\left(-\frac{1}{\lambda_j}
  + z \lambda_j\right)\prod_{1 \leq j < k \le n}
\abs{\lambda_{k}-\lambda_{j}}^{\beta}d\lambda_{1}\ldots d\lambda_{n},
\end{equation}
where $b$ is the same parameter that appears in~\eqref{eq:wdt_d} and
$m^{(\beta)}_n$ is a normalization constant such that
$M^{(\beta)}_n(0)=1$, namely
\begin{equation}
\label{eq:ncl}
m^{(\beta)}_n := \prod_{j=0}^{n-1}\frac{\Gamma\bigl(1 + (j+1)\beta/2\bigr)
  \Gamma\bigl(b -1 + (j - 2n+ 2)\beta/2\bigr)}{\Gamma\left(1 +
    \beta/2\right)}.
\end{equation}
This is a particular limit of Selberg's integral~\eqref{eq:selberg}
(see, \textit{e.g.},~\cite{For10}, Sec.~4.7.1).  The quantity
$M^{(\beta)}_{n}(\beta y/2)$ is the moment generating function of the
Wigner delay time~\eqref{eq:wdt_d}. 
\begin{theorem}
\label{th:wdt_ode}
Let $\beta \in \left \{1,2,4\right\}$ and define 
\begin{equation}
  \label{eq:false_c_gf}
  \xi^{(\beta)}_n(z) := \log M_n^{(\beta)}(z),
\end{equation}
the parameter 
\begin{equation}
\label{dnbeta}
d_{n}^{(\beta)} := 
\begin{cases}
\frac{n(n-1)}{(n+1)(n+2)}\frac{m^{(1)}_{n-2}m^{(1)}_{n+2}}%
   {\left(m^{(1)}_{n}\right)^{2}} &\text{if $\beta=1$,} \\
0 &\text{if $\beta=2$,}\\
\frac{n}{n+1}\frac{m^{(4)}_{n-1}m^{(4)}_{n+1}}%
  {\left(m^{(4)}_{n}\right)^{2}} &\text{if $\beta=4$}
\end{cases}
\end{equation}
and the polynomials 
\begin{subequations}
\begin{align}
  p^{(\beta)}_{1}(z) &= 2z-\left(b - \beta n - 2 +\beta\right)\left(b
    - \beta n\right),\\
  p^{(\beta)}_{2}(z) &= 4z^{2}+\beta z -\left(b-\beta n -6 +
    3\beta\right)\left(b - \beta n - 2 +\beta\right)z.
\end{align}
\end{subequations}
The function $\xi^{(\beta)}_n(z)$ is a solution of the ODE
\begin{multline}
\label{wigode}
\eta^{(\beta)}_{1,4}z^{3}\xi^{(\beta)\prime \prime \prime \prime}_{n}
+2\beta z^{2}\xi^{(\beta)\prime \prime \prime}_{n}
+\beta z\bigl(\xi^{(\beta)\prime}_{n}\bigr)^{2}+4\beta
z^{2}\xi^{(\beta)\prime \prime}_{n} \xi^{(\beta)\prime}_{n} \\
+6\eta^{(\beta)}_{1,4}z^{3}\bigl(\xi_{n}^{(\beta)\prime \prime}
\bigr)^{2}+
p^{(\beta)}_{2}(z)\xi^{(\beta)\prime \prime}_{n}
+p^{(\beta)}_{1}(z)\xi^{(\beta)\prime}_{n}+n(b-\beta n)\\
=\frac{12d^{(\beta)}_{n}}{\beta}z^{3}\exp\left(\xi^{(\beta)}_{n-i(\beta)}
+\xi^{(\beta)}_{n+i(\beta)}-2\xi^{(\beta)}_{n}\right).
\end{multline}
The quantities $\eta_{1,4}^{(\beta)}$ and $i(\beta)$ were defined in
Eq.~\eqref{def_par} and $b$ in~\eqref{eq:par_lag}. 
\end{theorem}

The approach outlined for conductance and shot noise can be applied to
the Wigner delay time too.  However, there are two features of the
\textit{j.p.d.f.}~\eqref{eq:wdt_d} that increase difficulties of the
asymptotic analysis substantially. First, the exponent~\eqref{eq:par_lag} is
proportional to the dimension $n$, while the parameters $\alpha$ and
$\delta$ that appear in the \textit{j.p.d.f.} of the transmission
eigenvalues are taken to be independent of $n$. Second, the singularity at the origin in the 
integrand of~\eqref{wignermgf} implies that $M_n^{(\beta)}(z)$ is not
analytic, therefore only a finite number of moments exist.  Thus, an
infinite power series is replaced by the asymptotic expansion
\begin{equation}
  \label{eq:cum_ex_wdt}
   \xi^{(\beta)}_n(z) = \sum_{l=1}^q \frac{K_l^{(\beta)}}{l!}\left(\frac{2
       z}{\beta}\right)^{l}  + o\left(z^{q}\right), \quad z \to 0,
\end{equation}
where the $K^{(\beta)}_l$'s are by definition the cumulants of the
random variable~\eqref{eq:w_td}. A brief inspection of the integral in
the right-hand side of~\eqref{wignermgf} shows that 
\begin{equation*}
q := \left
  \lfloor b - 2 - \beta \left(n-1\right) \right \rfloor,
\end{equation*}
where $\left \lfloor x \right \rfloor$ denotes the greatest integer
such that
\begin{equation*}
 \left \lfloor x \right \rfloor \le x, \quad x \in \mathbb{R}. 
\end{equation*}

Because of the increased technical difficulties, our asymptotic analysis focuses principally on the simpler case $\beta=2$. Our starting point is the following result, an immediate consequence of Theorem~\ref{th:wdt_ode}. 
\begin{corollary}
\label{co:ci_eq}
  Let $\beta=2$ and set $H_n(-z) = z\xi^{(\beta)\prime}_n(z)$.  Then 
$H_n(z)$ satisfies the non-linear second order ODE 
\begin{equation}
\label{eq:ci_eq}
\begin{split}
(zH_n'')^{2}& = 4H_n\bigl((H_n')^{2}-H_n'\bigr)\\
& \quad - \bigl(4z(H_n')^{2}-(4z+(b-2n)^{2})H_n' - 2n(b-2n)\bigr)H_n' +
n^2. 
\end{split}
\end{equation}
\end{corollary}
\begin{proof}
Substituting $H_n(z)$ into~\eqref{wigode} gives 
\begin{multline}
  z^{2}H_{n}'''+zH''_{n}+6z(H'_{n})^{2}-4H_{n}H'_{n}
  -\bigl(4z+(b-2n)^{2}\bigr)H'_{n}+2H_{n}-n(b-2n)=0. 
\end{multline}
This equation falls into one of \emph{the Chazy classes} of ODEs:
\begin{equation}
\label{chazeqn}
L'''+\frac{P'}{P}L''+\frac{6}{P}(L')^{2}-\frac{4P'}{P^{2}}LL'
+\frac{P''}{P^{2}}L^{2}+\frac{4Q}{P^{2}}L'-\frac{2Q'}{P^{2}}L
+\frac{2R}{P^{2}}=0,
\end{equation}
where $P,Q$ and $R$ are polynomials. 
Cosgrove \cite{Cos00,CS93} found a first integral of \eqref{chazeqn}:
\begin{multline}
  (L'')^{2}+\frac{4}{P^{2}}\bigl((P(L')^{2}+QL'+R)L'
    -(P'(L')^{2}+Q'L'+R')L\\
  +\frac{1}{2}(P''L'+Q'')L^{2}-\frac{1}{6}P'''L^{3}+C\bigr)=0,
\end{multline}
where $C$ is an integration constant. Setting
\begin{align*}
L(z) & = H_n(z), & P(z)& =z, & Q(z) &= -z-\frac{(b-2n)^{2}}{4}, \\
R(z)& = -\frac{n(b-2n)}{2}, & C &= - n^2 & &
\end{align*} 
completes the proof.  
\end{proof}
\begin{remark}
  When $\beta =2$, Eq.~\eqref{wigode} was studied by Osipov and
  Kanzieper~\cite{OK07} in the context of bosonic replica field
  theories, who realized that it can be reduced to Painlev\'e
  $\mathrm{III}^{\prime}$.  Chen and Its~\cite{CI10} made a detailed
  study of the partition function 
  $\frac{m_n^{(2)}}{n!}M_n^{(2)}(-z)$ and showed that 
  \begin{equation*}
    H_n(z)- \frac{z}{2} - \frac{n(b-n)}{2}
  \end{equation*}
  is the Jimbo-Miwa-Okamoto $\sigma$-function for 
  Painlev\'e $\mathrm{III}^{\prime}$. As far as we are aware, the present work 
  is the first to make explicit the connection between the distribution of the Wigner delay time and the Painlev\'e $\mathrm{III}^{\prime}$ transcendent.
\end{remark}

When $\beta=2$ Eq.~\eqref{wigode} can be used to obtain a
recursion relation for the leading order term of the
cumulants. Define
\begin{equation}
  \label{eq:limit_cum}
  p_l = \lim_{n \to \infty} \frac{n^{2l-2}K_l^{(2)}}{(l-1)!}, \quad l >0.
\end{equation}
as well as the generating function
\begin{equation}
  \label{eq:gf_c}
  F(z) = \sum_{i=1}^\infty p_iz^i.
\end{equation}

\begin{theorem}
\label{th:lim_cum_wdt}
The limit~\eqref{eq:limit_cum} exists and is an integer. It satisfies the
recurrence relation
\begin{equation}
\label{limwigrecintro}
(l+1)p_{l+1}=2(2l-1)p_{l}+2\sum_{i=0}^{l-1}(3i+1)(l-i)p_{i+1}p_{l-i},
\end{equation}
with initial condition $p_1=1$.  Furthermore, the generating function
$F(z)$ is the following:
\begin{equation}
\label{solwigode}
F(z) = 3/2-(3/2)\Omega(z)+2z\Omega(z)^{3}+3z\Omega(z)^{2}+2z\Omega(z), 
\end{equation}
where $\Omega(z)$ admits the power series expansion
\begin{equation}
  \label{eq:ser_exp}
  \Omega(z) = 1 + \sum_{k=1}^\infty \zeta_k z^k,
\end{equation}
where
\begin{equation}
  \label{eq:zeta_k}
  \zeta_k =
  \frac{4}{k}\sum_{i=0}^{k-1}\binom{2k-i-2}{k-1}
  \left(\sum_{p=0}^{i}\binom{2k}{p}\binom{2k}{i-p}2^{i+p}\right)(-3)^{k-1-i}.   
\end{equation}
\end{theorem}

\begin{remark}
  It is worth emphasizing that the coefficients $\zeta_k$, and hence 
  the leading order contribution to the cumulants, are \emph{integer
    numbers.} It would be interesting to understand the 
  physical reasons behind this combinatorial fact.  It is not true
  when $\beta=4$ (see Table~\ref{pie}).
\end{remark}

\begin{figure}
\centering
\includegraphics[width=5.3in]{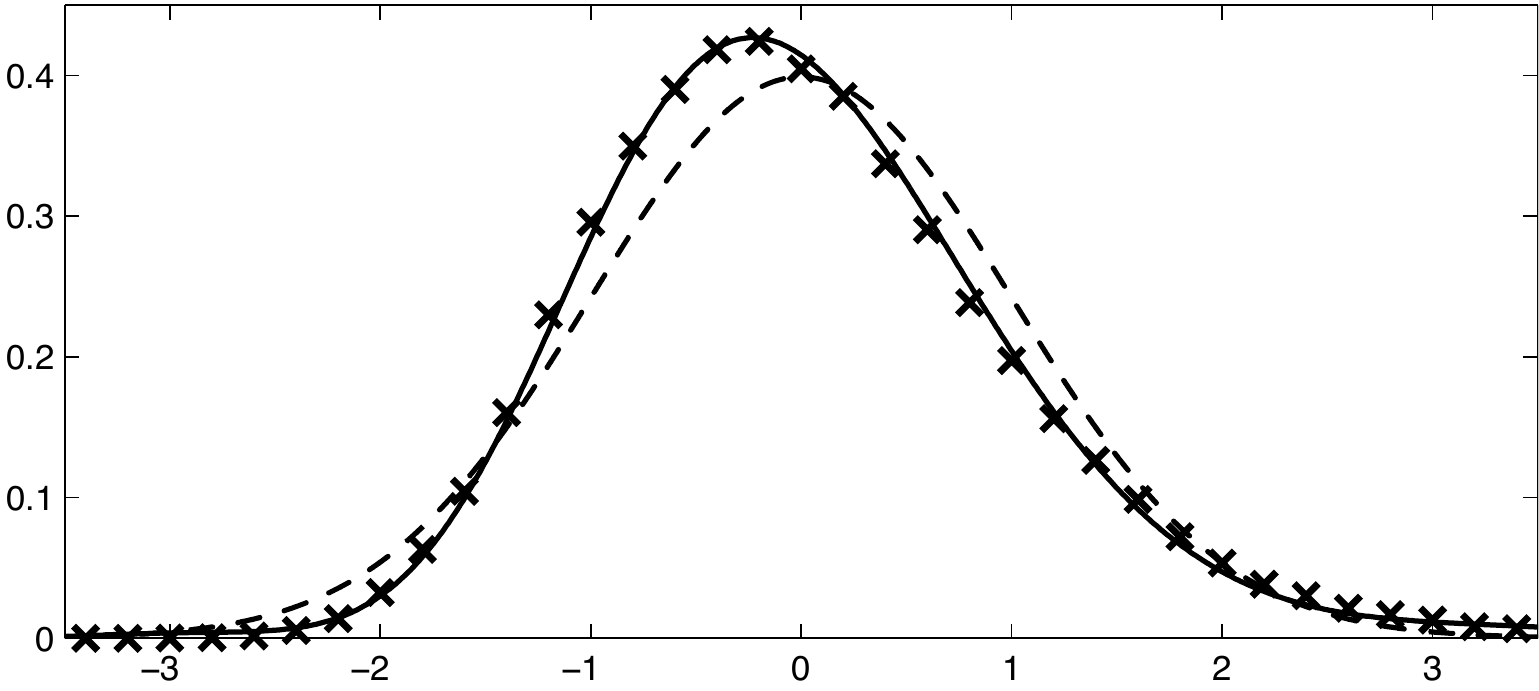}
\caption{\tiny{The probability density function of the Wigner delay time
    scaled and centered to have zero mean and unit variance. The
    crosses are a Monte-Carlo simulation for $n=20$ and $\beta=1$,
    computed using the $\beta$-Laguerre tri-diagonal ensemble with
    $100,000$ samples. The solid line is a five cumulant Edgeworth
    approximation, computed from our recurrence relations, and
    yielding remarkable agreement. The dashed line depicts the
    limiting Gaussian as $n \to \infty$.}}
\label{wigfig}
\end{figure}

By equating the first few powers of $z$ in the ODE~\eqref{wigode}, we obtain non-perturbative (finite-$n$) results for the lower order cumulants:
\begin{subequations}
\label{wigcumus}
\begin{align}
\label{wigmean}
K^{(\beta)}_{1} &= \frac{\beta}{2}\frac{n}{\omega} = 1, \\
\label{wigvar}
K^{(\beta)}_{2} &= \left(\frac{\beta}{2}\right)^{2}
\frac{n(2\omega+\beta
  n)}{\omega^{2}(2\omega+\beta)(\omega-1)} 
= \frac{4}{(n+1)(n\beta-2)}, \\
K^{(\beta)}_3 & = \left(\frac{\beta}{2}\right)^{3}
\frac{4n(2\omega+n\beta)(\omega+n\beta)}%
{\omega^{3}(2\omega+\beta)(\omega+\beta)(\omega-1)(\omega-2)}\notag \\ 
\label{wigskew}
& = \frac{96}{(n+1)(n+2)(n\beta-2)(n\beta-4)},
\end{align}
\end{subequations}
where $\omega = b-2 -\beta(n-1)$. Our formula (\ref{wigvar}) for the
variance can be extracted from earlier results in the quantum
transport literature \cite{FS97,FSS97}. When $\beta=2$, it has also
appeared in the context of wireless communications \cite{LTV03}. To
the best of our knowledge, our formula (\ref{wigskew}) for the third
cumulant has never appeared in the literature before.

From Eq.~\eqref{wigvar} it follows immediately
that the variance of the Wigner delay time satisfies the asymptotic formula
\begin{equation}
\label{limwigvar}
n^{2}K^{(\beta)}_{2} \sim \frac{4}{\beta}+\frac{4}{\beta n}
\left(\frac{2}{\beta}-1\right)+O(n^{-2}) \quad n \to \infty.
\end{equation}
In the literature on semiclassical approaches to this problem, it is a simple matter to see that Eq. \eqref{limwigvar} in addition to the first seven terms in the asymptotic expansion of Eq. \eqref{wigvar} are in agreement with semiclassical computations in \cite{KS07}. We believe it would be important to go further and obtain a semiclassical derivation of \eqref{wigvar} to all orders in $1/n$.


Beyond the third cumulant, the exact results become more
complicated. We conclude our discussion of the time delay by giving
the explicit expressions for the fourth cumulant in each symmetry
class: 
\begin{subequations}
\begin{align}
K^{(1)}_{4} &= \frac{96(53n^{2}-68n-156)}%
{(n-4)(n+1)^{2}(n-2)^{2}(n+3)(n+2)(n-6)},\\
K^{(4)}_{4} &= \frac{12(53n^{2}+34n-39)}%
{(n+2)(n+1)^{2}(n-1)(2n-3)(2n-1)^{2}(n+3)},\\
K^{(2)}_{4} &= \frac{12(53n^{2}-77)}{(n^{2}-1)^{2}(n^{2}-4)(n^{2}-9)}.
\end{align}
\end{subequations}
The higher order cumulants $K^{(\beta)}_{j}$ with $j>4$ can be obtained
systematically from our recurrence relations in Sec. 6, valid for any $n$ and
any $\beta \in \{1,2,4\}$. In Fig. \ref{wigfig} we used this to calculate an
Edgeworth series approximation to the exact distribution of $\tau_{W}$ for
$n=20$ and $\beta=1$. It shows that the cumulants accurately reproduce the
deviations from the limiting Gaussian as $n\to \infty$. Very recently, such
deviations were also successfully described using the Couloumb gas
approach~\cite{TM13}. 

\subsection{Shot Noise when $n=m$ and $\beta=2$}
\label{shotproperties}
We end the overview of our results by discussing two interesting
properties of the shot noise for symmetry classes with $\beta=2$ that
have remained unnoticed in the quantum transport literature. These are
immediate consequences of the definition of the moment generating
function~\eqref{intrepjointgen} and of previous results from Witte
\textit{et al.}~\cite{WFC00} and Forrester~\cite{For06a}.

When the two scattering leads in the cavity support the same number of
quantum channels, \textit{i.e.} $n=m$, the description of shot noise
simplifies considerably if $\beta=2$ because $\alpha =0$.  We first
discuss how the moment generating function can be expressed in terms
of the Painlev\'e V transcendent; then, we point out an
interesting 
identity between the probability distributions of the shot noise and
the conductance.
\subsubsection{Shot Noise and  Painlev\'e V}
\label{shotpainleve}
It is well known that gap probabilities in ensembles with $\beta=2$
are related to Painlev\'e transcendents.  In particular for the
Gaussian Unitary Ensemble we have the following.
\begin{proposition}[Witte, Forrester and Cosgrove~\cite{WFC00}]
  Consider the Gaussian Unitary Ensemble of random matrices and
  let $E_2(w)$ be the probability that the interval $(-\infty,-w)
  \cup (w,\infty)$ contains no eigenvalues.    Then, 
  \begin{equation}
    \label{eq:e2_r_WFC}
    R(w) := \frac12 \frac{d}{dw}\log E_2(w).
  \end{equation}
    satisfies the non-linear ODE 
    \begin{equation}
wR''+2R'=2w(w-h)-2h\sqrt{(R+wR')^{2}-4w^{2}(w-h)R-2nw^{2}(s-h)^{2}},
\end{equation}
where 
\begin{equation}
h=\sqrt{w^{2}-2R'}.
\end{equation}
\end{proposition}
The function $R(w)$ can be expressed in terms of Painlev\'e V
transcendents with appropriate boundary conditions~\cite{WFC00}.
 
It turns out that the gap probability $E_{2}(w)$ is directly related
to the moment generating function of the shot noise. More precisely,
we have 
\begin{equation}
  \label{eq:gap_shot}
  E_2(w) = C e^{-w^{2}n}(2w)^{n^{2}}\mathcal{M}_{n}^{(2)}(0,-4w^{2}),
\end{equation}
where $C$ is a multiplicative constant that does not affect the
definition~\eqref{eq:e2_r_WFC}.   The identity~\eqref{eq:gap_shot}
follows from
\begin{equation}
\label{gueintegral}
\begin{split}
  \mathcal{M}_{n}^{(2)}(0,-w) &= \frac{1}{c^{(2)}_{n}}
  \int_{[0,1]^{n}}\prod_{j=1}^{n}e^{wT_{j}(1-T_{j})}\prod_{1\le j <k
    \le
    n}\abs{T_k - T_j}^2 dT_1 \dotsm dT_n\\
  &=\frac{e^{wn/4}w^{-n^{2}/2}}{c^{(2)}_{n}}
 \int_{[-\sqrt{w}/2,\sqrt{w}/2]^{n}
 }\prod_{j=1}^{n}e^{-T_{j}^{2}}\prod_{1\le j <k \le
  n}\abs{T_k - T_j}^2 dT_1 \dotsm dT_n ,
\end{split}
\end{equation}
which is obtained by completing the square in the exponential and by
changing variables.

\subsubsection{Shot Noise and Conductance}
\label{reltocond}
There is a fascinating relation between the distributions of the shot 
noise and the conductance in non-superconductive and superconductive
quantum dots, respectively.
\begin{proposition}
  Let $P_n^{(2,0)}$, $G^{(2,-1)}_{N_1}$ and
  $G^{(2,1)}_{N_2}$ denote the random variables defined
  in~\eqref{eq:sec_con}, where
  \begin{equation}
    \label{eq:n1_n2}
    N_1 := \left \lfloor \frac{n + 1}{2}\right \rfloor \quad
    \text{and} \quad N_2 := \biggl \lfloor \frac{n}{2}\biggr \rfloor.
  \end{equation}
  The subscripts refer to the number of variables in the
  \textit{j.p.d.f.}~\eqref{eq:tr_eig_Aqd}, while the superscripts emphasize
  its dependence on the parameters $(\beta,\delta)$.   In
  addition we assume that $\alpha=0$.  We have 
   \begin{equation}
    \label{altland}
    P_{n}^{(2,0)} \overset{d}{\equiv}
   \frac{G^{(2,-1)}_{N_{1}}+G^{(2,1)}_{N_{2}}}{4},
\end{equation}
where the notation $\overset{d}{\equiv}$ denotes equivalence in
distribution. 
\end{proposition}
\begin{proof}
  The proof follows from a result of Forrester~\cite{For06a}, who
  exploited the symmetry of the integration interval in the right-hand
  side of~\eqref{gueintegral}, together with the evenness of the
  Gaussian weight. He thus obtained the
  identity 
  \begin{equation}
\label{productintegrals}
\begin{split}
  \mathcal{M}_{n}^{(2)}(0,-w) &= \frac{e^{wn/4}}{c_n^{(2)}}
      \int_{[0,1]^{N_{1}}}\prod_{j=1}^{N_{1}}T_{j}^{-1/2}e^{-wT_{j}/4}
    \prod_{1\le j <k \le
  N_1}\abs{T_k - T_j}^2 dT_1 \dotsm dT_{N_1} \\
  &\quad \times
  \int_{[0,1]^{N_{2}}}\prod_{j=1}^{N_{2}}T_{j}^{1/2}
e^{-wT_{j}/4}\prod_{1\le j <k \le
  N_2}\abs{T_k - T_j}^2 dT_1 \dotsm dT_{N_2}.
\end{split}
\end{equation}
Changing the variables $T_{j} \to 1-T_{j}$, and scaling $w \to 
4w$ gives 
\begin{equation}
  \mathcal{M}_{n}^{(2)}(0,-4w)=\mathcal{M}^{(2)}_{N_{1}}(-w)
\biggr \rvert_{\delta=-1}\mathcal{M}^{(2)}_{N_{2}}(-w)\biggr \rvert_{\delta=1}.
\end{equation}
Finally, Eq.~\eqref{altland} follows from the convolution
theorem.
\end{proof}
When $(\beta,\delta)=(2,-1)$, then the compact symmetric space of the
scattering matrix $S$ is of symmetry type DIII; this means that $S$ is
orthogonal and self-dual.  When $(\beta,\delta)=(2,1)$ the symmetric
space is of symmetry type CI; in this case $S$ is symplectic and
symmetric.  Note that matrices of symmetry type DIII are real, while
matrices of symmetry type CI have quaternion elements.  In both
instances time-reversal is preserved.  The Landauer-B\"uttiker
scattering matrix $S$ of the system on the right-hand side belongs to
the CUE; thus, interestingly, time-reversal is broken on the left-hand
side of Eq.~\eqref{altland} but not on the right-hand side.  

From a computational perspective, we note that the right-hand side
of~\eqref{altland} is far more amenable to calculation than the
left-hand side. As we shall see in Sec.~\ref{se:csnfn}, the cumulants of
$G^{(2,-1)}_{N_{1}}$ and $G^{(2,1)}_{N_{2}}$ can be determined from
the one-dimensional difference equation \eqref{reccond}, while to
obtain the cumulants of the shot noise we have to study the partial
recurrence relation \eqref{shotnoiserec}.

\section{Integrable Systems and Moment Generating Functions}
\label{se:int_the}
We can associate to a scattering observable, such as the Landauer
conductance or the Wigner delay time, a moment generating function
whose integral representation is a partition function of one
of the orthogonal, unitary or symplectic ensembles of random matrices,
depending on whether $\beta=1$, $\beta=2$ or $\beta=4$. In this
section we describe the general theory that provides a link between
such partition functions and certain integrable non-linear
differential equations.

\subsection{Matrix Integrals, $\tau$-Functions and Integrable
  Hierarchies}
\label{se:int_hier}
Fundamental to our approach is the connection between matrix integrals
and exactly solvable models, as established in a variety of papers,
\textit{e.g.,} \cite{MM90,GMMMO91,ASvM95,AvM95,TSY96,ASvM98}. These
works focus principally on the case $\beta=2$ and obtain
relations to the Toda lattice and KP hierarchies. In contrast, the
existence of integrable structures related to the $\beta=1,4$ matrix 
integrals \cite{AvM01a,ASvM02,AvM02} has received comparatively little
attention.
 
The starting point for establishing the
relationship with integrable hierarchies is the  \textit{deformed
  integral}
\begin{equation}
 \label{taugencond}
 \tau_{n}^{(\beta)}(\mathbf{t}) := \frac{1}{n!}\int_{[A,B]^{n}}
 \prod_{j=1}^{n}\rho(x_{j})\exp\left(\sum_{i=1}^{\infty}t_{i}x_{j}^{i}\right)
 \prod_{1\le j < k \le n}\abs{x_k - x_j}^{\beta}dx_1\dotsm dx_n,
 \end{equation}
 defined by introducing an infinite sequence of time variables 
 $\mathbf{t} = (t_{1},t_{2},t_{3},\ldots)$ into the integrands of
 Eqs.~\eqref{intrepjointgen} and~\eqref{wignermgf}.  For the mixed
 moment generating function of conductance and shot noise, 
 $[A,B]=[0,1]$ and the weight function is
\begin{equation}
\label{shotweight}
\rho_{z,w}(T) = T^{\alpha}(1-T)^{\delta/2}e^{-zT-wT(1-T)}.
\end{equation} 
For the Wigner delay time, we have $[A,B] = [0,\infty)$ and 
\begin{equation}
  \label{eq:wtd}
  \rho_z(\lambda) = \lambda^{-b}\exp\left(-\frac{1}{\lambda}
    +z\lambda\right),
\end{equation}
where $z<0$.  In the theory of integrable systems the integral~\eqref{taugencond} is
known as a $\tau$-function. 

When $\beta=2$ the integral \eqref{taugencond} simplifies 
considerably and can be written as a Hankel determinant, which can be
studied using orthogonal polynomials.  We have
\begin{equation}
  \tau^{(2)}_{n}(\mathbf{t}) = 
  \det\left(\int_{A}^{B}x^{i+j}
    \rho(x)e^{\sum_{i=1}^{\infty}t_{i}x^{i}}dx\right)_{0 \leq i,j \le n-1}.
\end{equation}

When $\beta=1$, however, no such nice structure exists.  In this case, 
for even $n$ the $\tau$-function~\eqref{taugencond} becomes the Pfaffian 
of a skew symmetric matrix (the square root of its
determinant). Namely, (see, \eg,~\cite{Meh04})
\begin{equation}
\label{pfaff}
\tau_{n}^{(1)}(\mathbf{t}) = \Pf\left(\frac{1}{2}\int_{A}^{B}
\int_{A}^{B}y^{k}x^{l}\sgn(y-x)
e^{\sum_{i=1}^{\infty}t_{i}(y^{i}+x^{i})}\rho(y)
\rho(x)dydx\right)_{0 \leq k,l \leq n-1}.
\end{equation}

Although we shall restrict our attention to the weights defined in 
\eqref{shotweight} and \eqref{eq:wtd}, the identity \eqref{pfaff} holds 
independently of these choices or the deformation parameters
$\mathbf{t}$. Formula \eqref{pfaff} is based on an identity of de
Bruijn for integrating determinants:
\begin{multline}
\label{debruyn}
\frac{1}{n!}\int_{[A,B]^{n}}\det\left(F_{i}(y_{1}), 
  G_{i}(y_{1}), \dotsc, F_{i}(y_{n}), G_{i}(y_{n})\right)_{0
  \leq i \leq 2n-1}dy_{1}\dotsm dy_{n}\\ =
\Pf\left(\int_{A}^{B}(G_{i}(y)F_{j}(y)-F_{i}(y)G_{j}(y))dy\right)_{0
  \leq i,j \leq 2n-1},
\end{multline}
where $F_i$ and $G_i$, $i=0,\dotsc,n-1$, are two sets of integrable
functions. In the left-hand side of \eqref{debruyn}, the commas
separate columns and the index $i$ labels the rows.

By studying the evolution of the moment matrix \eqref{pfaff} in the
extended $\mathbf{t}$-space, Adler and van Moerbeke
\cite{AvM01a,AvM02,ASvM02} discovered a \textit{bi-linear identity}
that applies to $\tau$-functions that can be expressed as Pfaffians:
\begin{multline}
\label{bilid}
\oint_{\mathcal{C}_{\infty}}
\tau^{(1)}_{n}\left(\mathbf{t}-\left[\xi^{-1}\right]\right)
\tau^{(1)}_{m+2}\left(\mathbf{t}'+\left[\xi^{-1}\right]\right)
e^{\sum_{i=1}^{\infty}(t_{i}-t_{i}')
  \xi^{i}}\xi^{n-m-2}\frac{d\xi}{2\pi i}\\
=-\oint_{\mathcal{C}_{0}}\tau^{(1)}_{m}(\mathbf{t}'-[\xi])
\tau^{(1)}_{n+2}(\mathbf{t}+[\xi])
e^{\sum_{i=1}^{\infty}(t_{i}'-t_{i})\xi^{-i}}\xi^{n-m}\frac{d\xi}{2
  \pi i},
\end{multline} 
where $n$ and $m$ are even, and $\mathcal{C}_{\infty}$ or
$\mathcal{C}_{0}$ are small circles in the complex plane enclosing the
point $\xi=\infty$ and $\xi=0$, respectively. The notation
$[\xi]$ refers to the infinite vector
$(\xi,\xi^{2}/2,\xi^{3}/3,\ldots)$. 

The identity (\ref{bilid}) can be used to generate an infinite
sequence of integrable hierarchies of PDEs involving
$\tau^{(1)}_{n}(\mathbf{t})$. Evaluating the residues on both sides of
\eqref{bilid} and expanding up to terms linear in
$(\mathbf{t}-\mathbf{t}')$, one obtains a hierarchy of PDEs known as the
\textit{Pfaff-KP hierarchy;} its first non-trivial member is the
\textit{Pfaff-KP equation:}
\begin{equation}
\label{pkpequation}
\left(\frac{\partial^{4}}{\partial t_{1}^{4}}
+3\frac{\partial^{2}}{\partial t_{2}^{2}}
-4\frac{\partial^{2}}{\partial t_{1}\partial t_{3}}\right)
\log \tau^{(1)}_{n}
+6\left(\frac{\partial^{2}}{\partial t_{1}^{2}}
\log \tau^{(1)}_{n}\right)^{2}
=12\frac{\tau^{(1)}_{n-2}\tau^{(1)}_{n+2}}{\bigl(\tau^{(1)}_{n}\bigr)^{2}}.
\end{equation}
When Eq. \eqref{pkpequation} was introduced in RMT, it was used
to characterize gap formation probabilities in orthogonal and
symplectic ensembles~\cite{AvM01a}. To the best of our knowledge, this
article is the first to apply \eqref{pkpequation} in a physical
context for over a decade. 

A similar discussion holds when $\beta=4$. One can show that
$\tau^{(4)}_{n}(2\mathbf{t})$ satisfies the same equation as
$\tau^{(1)}_{n}(\mathbf{t})$, except for one small difference: in the
right-hand side the sub-indices $n\pm 2$ are replaced by $ n\pm 1$
and there is no
parity 
restriction on $n$.

Although our emphasis is on the $\beta=1,4$ ensembles, we point out
that when $\beta=2$ there is a simpler and more widely known analogue
of \eqref{pkpequation} usually referred to as the \textit{KP
  equation:}
\begin{equation}
\label{kpequation}
\left(\frac{\partial^{4}}{\partial t_{1}^{4}}
+3\frac{\partial^{2}}{\partial t_{2}^{2}}
-4\frac{\partial^{2}}{\partial t_{1}\partial t_{3}}\right)\log \tau^{(2)}_{n}
+6\left(\frac{\partial^{2}}{\partial t_{1}^{2}}\log \tau^{(2)}_{n}\right)^{2}=0,
\end{equation}
which is the first member of the KP hierarchy.

It is worth emphasizing that the right-hand side of
Eq.~\eqref{pkpequation} is not zero as well as containing
$\tau^{(1)}_{n \pm 2}(\mathbf{t})$ ($\beta=1$) or $\tau^{(4)}_{n\pm
  1}(2\mathbf{t})$ ($\beta=4$). This constitutes an additional
complication compared to the symmetry classes with
$\beta=2$. 

Equation~\eqref{pkpequation} is the starting point to obtain the
differential equations in Theorems~\ref{pde:mix_csn}, \ref{ode:con}
and~\ref{th:wdt_ode}; in turn such differential equations lead to
recurrence relations for the cumulants. Deriving these differential
equations requires a combination of techniques pioneered by Adler and
van Moerbeke~\cite{AvM01a}.  More precisely we need extra information
on the properties of $\tau$-functions that is provided by the fact 
that~\eqref{taugencond} satisfies the \textit{Virasoro constraints.}
These are an infinite sequence of linear differential equations with a
specific algebraic structure, which we discuss next.

\subsection{Virasoro Constraints and $\beta$-Integrals}
\label{sse:vis_con}

The Virasoro constraints were first applied to unitary
matrix models ($\beta=2$) in the context of string
theory~\cite{MM90,GMMMO91}. In the unitary case, Osipov and Kanzieper~\cite{OK09} were the
first to apply these techniques to quantum
transport and used (\ref{kpequation}) to derive a PDE satisfied by
$\log\mathcal{M}_{n}^{(2)}(z,w)$. The fact that the Virasoro constraints apply for general $\beta>0$ appears not to have been applied to problems in quantum transport until now. 

The integrals~\eqref{taugencond} are well-defined for any real 
$\beta >0$.  Usually, when we want to emphasize this property, we
refer to them as $\beta$\textit{-integrals.} The Virasoro constraints 
we need follow from a particular change of integration variables. More generally, 
let $f(x)$ and $g(x)$ be real analytic in a neighbourhood of the
origin and supported on an interval $[A,B] \subseteq
\mathbb{R}$.\footnote{In general the domain of integration
  of~\eqref{taugencond} is $E^n$, where $E$ is a disjoint union of
  intervals contained in $[A,B]$.  In the cases we study, $E=[A,B]$, 
  which slightly simplifies the analysis.} 
In addition impose the boundary conditions
\begin{equation}
\label{eq:b_c}
V'(x) :=\frac{g(x)}{f(x)} 
 = \frac{\sum_{i=0}^{\infty}b_{i}x^{i}}{\sum_{i=0}^{\infty}a_{i}x^{i}},
\quad \lim_{x \to A,B} f(x)\rho(x)x^k = 0, \quad \text{for all $k>0$.}
\end{equation}
The $\tau$-functions~\eqref{taugencond} belong to the class of 
$\beta$-integrals with weights of the form
\begin{equation}
\label{eq:constr}
\rho(x) := e^{-V(x)}, \quad x \in [A,B] \subseteq \mathbb{R}.
\end{equation}
It is straightforward to check that the boundary
conditions~\eqref{eq:b_c} are satisfied by both
weights~\eqref{shotweight} and~\eqref{eq:wtd}.

Let us introduce the differential operators
\begin{align*}
J_{k}^{(1)} & := \begin{cases}
\frac{\partial}{\partial t_{k}} & k>0,\\
\frac{1}{\beta}(-k)t_{-k} & k<0,\\
0 & k=0,
\end{cases}\\
J_{k}^{(2)}& := \sum_{i=1}^{k-1}
\frac{\partial^{2}}{\partial t_{i} \partial t_{k-i}}+\frac{2}{\beta}
\sum_{i=1}^{\infty}it_{i}\frac{\partial}{\partial t_{i+k}},
\end{align*}
where $t_1,\dots,t_k$ are the first $k$ elements of the infinite
vector $\mathbf{t}=(t_1,t_2,t_3,\dotsc)$. Define
\begin{subequations}
\label{eq:v_op}
\begin{align}
\label{eq:v_op_1}
  \mathbb{J}_{k,n}^{(1)} & := J_{k}^{(1)}+n\delta_{0,k},\\
\label{eq:v_op_2}
  \mathbb{J}_{k,n}^{(2)} &:= 
\left(\frac{\beta}{2}\right)J_{k}^{(2)}
+\Biggl(n\beta+(k+1)\left(1-\frac{\beta}{2}\right)\Biggr)
J_{k}^{(1)}+n\left((n-1)\frac{\beta}{2}+1\right)\delta_{0,k},
\end{align}
\end{subequations}
as well as  the Virasoro operators 
\begin{equation}
\label{eq:vir_op}
  \mathcal{V}_{k} : =\sum_{i=0}^{\infty}
\left(a_{i}\, \mathbb{J}_{k+i,n}^{(2)}-b_{i}\,
  \mathbb{J}_{k+i+1,n}^{(1)}\right), \quad k \ge -1,
\end{equation}
where the parameters $a_i$ and $b_i$, $i=0,1,2,\dotsc,$ are the
coefficients of the Taylor expansion of $f(x)$ and $g(x)$ defined in
Eq.~\eqref{eq:b_c}.  Imposing the invariance of \eqref{taugencond}
under the transformation
\begin{equation}
\label{substit}
x_{j} \to x_{j}+\epsilon f(x_{j})x_{j}^{k+1}, \quad j=1,\dotsc,n.
\end{equation}
one can show that (see Appendix~$5$ of \cite{AvM01a}) 
\begin{equation}
  \label{tauannihilation}
  \mathcal{V}_k\tau_n(\mathbf{t}) = 0, \quad k \ge -1.
\end{equation}
We shall refer to these identities as \textit{Virasoro constraints.} 
A more abstract proof of Eq.~\eqref{tauannihilation} is based on the
fact that the $\beta$-integrals~\eqref{taugencond} are the fixed
points of vertex operators \cite{AvM01a}. 

It is worth emphasizing that the differential
operators~\eqref{eq:v_op} form an algebra, although this fact will not
enter directly in our discussion. For all $n \in \mathbb{Z}_+$ the
operators $\mathbb{J}^{(1)}_{k,n}$ and $\mathbb{J}^{(2)}_{k,n}$ obey a
Virasoro and Heisenberg algebra
\begin{align*}
  \left[\mathbb{J}^{(1)}_{k,n},\mathbb{J}^{(1)}_{\ell,n}\right] & =
  \frac{k}{\beta}\delta_{k,-\ell},\\
    \left[\mathbb{J}^{(2)}_{k,n},\mathbb{J}^{(1)}_{\ell,n}\right] & =
  - \ell \, \mathbb{J}^{(1)}_{k +\ell,n} +c' k(k+1)\delta_{k,-\ell}, \\
  \left[ \mathbb{J}^{(2)}_{k,n},\mathbb{J}^{(2)}_{\ell,n}\right] & =
  \left(k-\ell\right)\mathbb{J}^{(2)}_{k+\ell,n} + c\left(\frac{k^3
      -k}{12}\right)\delta_{k,-\ell}
\end{align*}
with central charge
\begin{equation*}
  c = 1 - 6\Biggl(\left(\frac{\beta}{2}\right)^{1/2} -
  \left(\frac{\beta}{2}\right)^{-1/2}\Biggr)^2 \quad \text{and} \quad 
c' = \frac{1}{\beta}-\frac12.
\end{equation*}

\subsection{Virasoro Constraints for Joint Conductance and Shot
  Noise}
\label{ss:vc_csn}
To put the theory of the matrix integrals described in 
Secs.~\ref{se:int_hier} and~\ref{sse:vis_con} into our context, let us
look at the $\tau$-function obtained by deforming the joint moment
generating function~\eqref{intrepjointgen}.

Since only linear and quadratic terms in the transmission eigenvalues
$T_j$ appear in the weight~\eqref{shotweight}, we only need the first
two Virasoro constraints. The logarithmic derivative of the
weight~\eqref{shotweight}
\begin{equation*}
\begin{split}
\frac{f(T)}{g(T)} & =- \frac{d}{dT}\log \rho_{z,w}(T) \\
& = \frac{\alpha-(\alpha+\delta/2+z+w)T+(z+3w)T^{2}-2wT^{3}}{T(T-1)}
\end{split}
\end{equation*}
gives the coefficients of the Taylor expansions of  $f(T)$ and $g(T)$:
\begin{alignat*}{4}
  a_0 & = 0, &\quad a_1 & =-1, & \quad a_2 & = 1, &\quad  a_3 & =0,  \\
 b_0 & = \alpha, & b_{1} &=-(\alpha+\delta/2 +z+w), & 
b_2 & = z + 3w, &  b_3  & = -2w,\\
a_i&= 0, &  b_i & = 0 \quad \text{if $i>3$.} && && 
\end{alignat*}
The first two Virasoro operators are: 
\begin{subequations}
\begin{equation}
\begin{split}
\label{shotvira1}
\mathcal{V}_{-1} &= \sum_{i=1}^{\infty}it_{i}
\left(\frac{\partial}{\partial t_{i+1}}
\ -\frac{\partial}{\partial t_{i}}\right)+
\Bigl(\alpha+z+\delta/2+w+\bigl(n\beta+2-\beta\bigr)\Bigr)
\frac{\partial}{\partial t_{1}}\\
&\quad -(z+3w)\frac{\partial}{\partial t_{2}}+2w\frac{\partial}{\partial
  t_{3}}-\Bigl(\alpha n+n((n-1)\beta/2+1)\Bigr) 
\end{split}
\end{equation}
and
\begin{equation}
\label{shotvira2}
\begin{split}
\mathcal{V}_{0} &= \sum_{i=1}^{\infty}it_{i}
\left(\frac{\partial}{\partial t_{i+2}}
-\frac{\partial}{\partial t_{i+1}}\right)+\frac{\beta}{2}
\frac{\partial^{2}}{\partial t_{1}^{2}}
-\Bigl(\bigl(n\beta+2-\beta\bigr)+\alpha\Bigr)
\frac{\partial}{\partial t_{1}}\\
&\quad +\Bigl(\alpha+z+\delta/2+w+\bigl(n\beta+3-3\beta/2\bigr)\Bigr)
\frac{\partial}{\partial t_{2}}\\
& \quad -(z+3w)\frac{\partial}{\partial
  t_{3}}+2w\frac{\partial}{\partial t_{4}}. 
\end{split}
\end{equation}
\end{subequations}

The Virasoro operators associated to the deformation of the moment
generating function of the conductance are obtained simply by setting
$w=0$ in Eqs.~\eqref{shotvira1} and~\eqref{shotvira2}.

\section{Conductance and Shot Noise for Finite-$n$}
\label{se:csnfn}

The Virasoro constraints allow us to evaluate the
Pfaff-KP 
equation~\eqref{pkpequation} and the KP equation~\eqref{kpequation} at
$\mathbf{t}=0$. This gives the differential 
equations~\eqref{diffeqshot} and~\eqref{diffeqcond}, thus proving
Theorems~\ref{pde:mix_csn} and~\ref{ode:con}. Then, by inserting the
series expansions~\eqref{mixcumugen} into these differential equations
we obtain recurrence relations, which allow both non-perturbative and
asymptotic analysis of the cumulants as $n \to \infty$.

\subsection{Differential-Difference Equations: the Conductance}
\label{ss:dd_eq}
As discussed in Sec~\ref{ss:c_sn_mr} the cumulants of the conductance
$\kappa_l^{(\beta)}$ provide the initial conditions for the
recurrence relation for the mixed cumulants $\kappa_{l,k}^{(\beta)}$.
Therefore, we shall start with the proof of Theorem~\ref{ode:con}.

Let us write down the $\tau$-function associated to the
weight~\eqref{shotweight} explicitly
\begin{equation}
\label{eq:tf_jmcsn}
\begin{split}
  \tau_{n}(\mathbf{t};z,w) &:= \frac{1}{n!}\int_{[0,1]^{n}}
\prod_{j=1}^{n}T_{j}^{\alpha}(1-T_{j})^{\delta/2}\\
  &\quad \times
  \exp\left(-zT_{j}-wT_{j}(1-T_{j})+\sum_{i=1}^{\infty}T_{j}^{i}t_{i}\right)
\prod_{1 \le j<k\le n}\abs{T_k - T_j}^{\beta}dT_1 \dotsm T_n.
\end{split}
\end{equation}
We now study the $\tau$-function $\tau_n(\mathbf{t};z) :=
\tau_n(\mathbf{t};z,0)$ as well as the moment generating function
\begin{equation}
  \label{eq:mgf_c}
  \mathcal{M}_n^{(\beta)}(z) := \frac{n!}{c_n^{(\beta)}}
  \tau_n^{(\beta)}(\mathbf{t};z)\biggr \rvert_{\mathbf{t}=\mathbf{0}}.
\end{equation}

First note that the definition~\eqref{eq:tf_jmcsn}
and~\eqref{taugencond} imply the relation
\begin{equation}
\label{eq:zt1d}
\frac{\partial}{\partial t_{1}}\tau^{(\beta)}_{n}(\mathbf{t};z) 
= -\frac{\partial}{\partial z}\tau^{(\beta)}_{n}(\mathbf{t};z).
\end{equation}
Eqs.~\eqref{eq:mgf_c} and~\eqref{eq:zt1d} provide the projections at
$\mathbf{t}=\mathbf{0}$ of the right-hand side of~\eqref{pkpequation} (for both
$\beta=1$ and $\beta=4$) as well as of the partial derivatives with
respect to $t_1$ in the left-hand side. Therefore, in order to
complete the proof of Theorem~\ref{ode:con} we need to express the
partial derivatives
\begin{equation}
\label{missingpds}
\frac{\partial^{2}}{\partial
  t_{2}^{2}}\log\tau^{(\beta)}_{n}(\mathbf{t};z)
\biggr \rvert_{\mathbf{t}=\mathbf{0}}\quad \text{and} \quad 
 \frac{\partial^{2}}{\partial t_{1}\partial t_{3}}
\log\tau^{(\beta)}_{n}(\mathbf{t};z)\biggr \rvert_{\mathbf{t}=\mathbf{0}}
\end{equation}
in terms of the derivatives of $\sigma_n^{(\beta)}(z) = \log
\mathcal{M}^{(\beta)}_{n}(z)$. This is a standard calculation
involving the Virasoro operators~\eqref{shotvira1}
and~\eqref{shotvira2}, whose details are provided in
Appendix~\ref{virconapp}.

\subsection{Non-Perturbative Recurrence Relation for the Cumulants of
  the Conductance}
\label{ss:np_con_cum}

The ODE~\eqref{diffeqcond} combined with the Taylor expansion
\begin{equation}
  \label{eq:ts_cond}
  \sigma^{(\beta)}_n(z) := \sum_{l=1}^\infty \frac{(-1)^l
   \kappa_l^{(\beta)}}{l!}z^l
\end{equation}
provides a non-perturbative recurrence relation that the cumulants
$\kappa_l^{(\beta)}$ satisfy. In Sec.~\ref{as_an} we shall study its
asymptotic limit as $n \to \infty$.

\begin{lemma}
\label{le:cc_rr}
The cumulants $\kappa_l^{(\beta)}$ satisfy the recurrence relation 
\begin{multline}
\label{reccond}
A(l)\kappa^{(\beta)}_{l+1}+\eta_{1,4}^{(\beta)}\sum_{i=0}^{l-1}
\binom{l}{i}\kappa^{(\beta)}_{i+1}\kappa^{(\beta)}_{l-i}B(i,l)
-l(2l-1)(\alpha -\delta/2 +
\beta n)\kappa^{(\beta)}_{l}\\
-l(l-1)(l-2)\kappa^{(\beta)}_{l-1} 
=(12b_{n}^{(\beta)}/\beta)l(l-1)(l-2)\mu^{(\beta)}_{l-3}, \quad
l\ge 3 
\end{multline}
with initial conditions
\begin{subequations}
\label{icreccond}
\begin{align}
\label{icreccond_1}
\kappa^{(\beta)}_{1} &= \frac{n\bigl(\alpha+1+\beta(n-1)/2\bigr)}%
 {\alpha+\delta/2+2+\beta(n-1)},\\
\label{icreccond_2}
	\kappa^{(\beta)}_{2} &=
\frac{1}{4}\frac{n\bigl(2\alpha+2+\beta(n-1)\bigr)
\bigl(\delta+2+\beta(n-1)\bigr)}%
{\bigl(\alpha+\delta/2+2+\beta(n-1)\bigr)^{2}\bigl(\alpha+\delta/2+3
+\beta(n-1)\bigr)} \\
& \quad \times \frac{\bigl(\delta+2\alpha+4+\beta(n-2)\bigr)}%
{\bigl(2\alpha+\delta+4+\beta(2n-3)\bigr)},\notag\\
\label{icreccond_3}
\kappa^{(\beta)}_{3} &= \frac{2\kappa_{2}^{(\beta)}(\delta/2
-\alpha+2\beta\kappa^{(\beta)}_{1}-\beta n)}%
{\bigl(\alpha+\delta/2+4+\beta(n-1)\bigr)\bigl(\alpha+\delta/2
+2+\beta(n-2)\bigr)}.
\end{align}
\end{subequations}

The parameter $b_n^{(\beta)}$ was defined in Eq.~\eqref{bnbeta} and
the coefficients $A(l)$ and $B(i,l)$ are polynomials in $i$ and $l$:
\begin{subequations}
\label{eq:pAB}
\begin{align}
  \label{eq:pA}
  A(l) & := \eta^{(\beta)}_{1,4}l(l-1)(l-2)+\beta
  l(2l-1) - (6-3\beta)(\alpha+\delta/2+\beta n+2-\beta)l \\
  & \quad \; -(\alpha+\delta/2+\beta n)
  (\alpha+\delta/2+\beta n+2-\beta)(l+1),\notag \\
\label{eq:pB}
B(i,l) &:= (2-\chi_{1,2}^{(\beta)})(l-i)\Bigl((l-i)(6i-2)+3\Bigr)
+(\chi_{1,2}^{(\beta)}-1)(l-i)^{2}(6i+2),
\end{align}
\end{subequations}
where 
\begin{equation*}
  \chi_{1,2}^{(\beta)} := \begin{cases} 1 & \text{if $\beta = 1,4$,}\\
                                          2 & \text{if $\beta=2$.}
                            \end{cases}
\end{equation*}
The coefficients $\mu_{l}^{(\beta)}$, $l=0,1,2\dotsc$ are the
``reduced moments'', which can be obtained from the recurrence relation
\begin{equation}
\label{recmomcum}
\mu^{(\beta)}_{l} =\sum_{j=0}^{l-1}\binom{l-1}{j}r^{(\beta)}_{l-j}
\mu^{(\beta)}_{j}, \qquad l=1,2,\dotsc,
\end{equation}
where
\begin{equation}
\label{redcumdef}
r^{(\beta)}_{l} := \kappa^{(\beta)}_{l}\Bigr \rvert_{n \to n-i(\beta)}
+\kappa^{(\beta)}_{l}\Bigr \rvert_{n \to
  n+i(\beta)}-2\kappa^{(\beta)}_{l},
\end{equation}
and by convention we have set $\mu_0=1$. The integer $i(\beta)$ was
introduced in Eq.~\eqref{def_par}.
\end{lemma}

This lemma can be proved by inserting the expansion~\eqref{eq:ts_cond} into
the ODE~\eqref{diffeqcond} and performing straightforward, although
tedious, calculations.  The initial conditions~\eqref{icreccond}
are obtained by equating the coefficients of the first three powers of
$z$ on both sides of the ODE.  The right-hand side of
Eq.~\eqref{diffeqcond} contributes only for powers $z^l$ such that $l
>3$ and is expanded using the identity
\begin{equation}
\label{eq:rm_exp}
   \exp\left(\sum_{l=1}^{\infty}
\frac{r^{(\beta)}_{l}z^{l}}{l!}\right) 
=  1+ \sum_{l=1}^{\infty}\frac{\mu^{(\beta)}_{l}z^{l}}{l!}.
\end{equation}

Since the coefficient $\mu^{(\beta)}_{l-3}$ on the right-hand side
of~\eqref{reccond} involves only the first $l-3$ reduced cumulants,
the difference equation~\eqref{reccond} and the initial
conditions~\eqref{icreccond} define the cumulants of conductance for
systems with $\beta=1,4$ uniquely. It also presents an efficient
algorithm for their computation.

\begin{remark}
  Formulae~\eqref{icreccond} generalise previous results~\cite{SS06,
    SWS07, SSW08, KSS09} to the symmetry classes $\delta \neq 0$. The
  right-hand side of the differential equation~\eqref{diffeqcond} does
  not contribute to the recurrence relation until one equates the
  coefficients of the third power
  of 
  $z$.  In other words, it does not contribute to the first three
  cumulants. It is possible to verify using only the Virasoro
  constraints that the formulae in~\eqref{icreccond} are valid for
  arbitrary $\beta>0$. 
  \end{remark}


\begin{remark}
  By setting $l=3$ in \eqref{reccond} we obtain  exact formulae for the
  fourth cumulant of conductance in terms of the first three.  For
  example, when $\beta=1,2$ we have 
\begin{subequations}
\label{4thcumu}
\begin{align}
\label{4thcumu1}
\kappa^{(1)}_{4} &= 3\left(\frac{-2\kappa^{(1)}_{2}
+10\kappa^{(1)}_{1}\kappa^{(1)}_{3}+22\bigl(\kappa^{(1)}_{2}\bigr)^{2}
+5\kappa^{(1)}_{3}(\delta/2-\alpha-n)-24b^{(1)}_{n}}%
{(\alpha+\delta/2+n+4)(4\alpha+2\delta+4n-3)}\right), \\
\label{4thcumu2}
  \kappa^{(2)}_{4} &= \frac{3}{4}\left(\frac{-2\kappa^{(2)}_{2}
+20\kappa^{(2)}_{1}\kappa^{(2)}_{3}+32\bigl(\kappa^{(2)}_{2}\bigr)^{2}
+5\kappa^{(2)}_{3}(\delta/2-\alpha-2n)}{(\alpha+\delta/2+2n)^{2}-9}\right)  .
\end{align}
\end{subequations}
When $\beta=1$, formula~\eqref{4thcumu1} involves the constant
$b^{(1)}_{n}$, whose value is given explicitly in Eq.~\eqref{bn1form}.
At first sight it seems that both cumulants~\eqref{4thcumu1}
and~\eqref{4thcumu2} should be of order $O(n^{-2})$ as $n \to \infty$,
since $\kappa^{(\beta)}_{j}$ is $O(n^{2-j})$ for $j=1,2$ and
$\kappa^{(\beta)}_{3}$ is $O(n^{-3})$. A closer inspection reveals
that \eqref{4thcumu2} contains mutually cancelling terms which lead to
$\kappa^{(2)}_{4}$ being $O(n^{-4})$. When $\beta=1$ these
cancellations are suppressed by the constant $b^{(1)}_{n}$ causing
$\kappa^{(1)}_{4} = O(n^{-3})$. From the recursion
relation~\eqref{reccond}, one easily sees that this ``staircase''
effect persists in all higher cumulants of the conductance. When
$\beta=4$ we observe a similar behaviour to the case $\beta=1$.
\end{remark}
\subsection{Differential-Difference Equation: Joint Conductance and
  Shot Noise}
\label{ss:mdd_eqw}
The aim of the this subsection is to prove Theorem~\ref{pde:mix_csn}
by deriving the non-linear PDE~\eqref{diffeqshot}.  The strategy, as in 
the proof of Theorem~\ref{ode:con}, is to combine the Pfaff-KP
equation~\eqref{pkpequation} satisfied by the
$\tau$-function~\eqref{eq:tf_jmcsn} with the Virasoro constraints.

First  observe the trivial identities
\begin{subequations}
  \label{eq:id_mcsn}
  \begin{align}
    \label{shotsimprel}
    \tau^{(\beta)}_{n}(\mathbf{t};z,w)\Bigr\rvert_{\mathbf{t}=0} & =
    \frac{c^{(\beta)}_{n}}{n!}\mathcal{M}^{(\beta)}_{n}(z,w),\\
    \label{t1dirshot}
\frac{\partial}{\partial t_{1}}
  \tau^{(\beta)}_{n}(\mathbf{t};z,w) & = 
   -\frac{\partial}{\partial z}\tau^{(\beta)}_{n}(\mathbf{t};z,w),\\
   \label{t2dirshot}
   \frac{\partial}{\partial t_{2}}\tau^{(\beta)}_{n}(\mathbf{t};z,w)&  =
   \frac{\partial}{\partial
     w}\tau^{(\beta)}_{n}(\mathbf{t};z,w)-\frac{\partial}{\partial
     z}\tau^{(\beta)}_{n}(\mathbf{t};z,w).
  \end{align}
\end{subequations}
Therefore, the partial derivatives of any order with respect to $t_{1}$ and $t_{2}$ of 
$\tau_n^{(\beta)}(\mathbf{t};z,w)$ can be expressed in terms of the
partial derivatives of $\mathcal{M}^{(\beta)}_{n}(z,w)$. It follows
that the only missing term in the Pfaff-KP equation
\eqref{pkpequation} is 
\begin{equation}
\label{t1t3need}
\frac{\partial^{2}}{\partial t_{1} \partial t_{3}}
\log \tau^{(\beta)}_{n}(\mathbf{t},z,w)\biggr\rvert_{\mathbf{t}=0}.
\end{equation}

Writing the constraint \eqref{shotvira1} in terms of $\log\tau_{n}(\mathbf{t},z,w)$ and taking the derivative with respect to $t_{1}$, the projection at $\mathbf{t}=\mathbf{0}$ yields the relation 
\begin{multline}
\label{t1t3dirshot}
\left(\frac{\partial}{\partial t_{2}}
-\frac{\partial}{\partial t_{1}}
+(\alpha+\delta/2
+z+w+n\beta+2-\beta)\frac{\partial^{2}}{\partial t_{1}^{2}}\right.\\
\left.\left.-(z+3w)\frac{\partial^{2}}{\partial t_{1} \partial t_{2}}
+2w\frac{\partial^{2}}{\partial t_{1}\partial t_{3}}\right)
\log \tau(\mathbf{t};z,w)\right\rvert_{\mathbf{t}=0}=0.
\end{multline}
Using the relations~\eqref{eq:id_mcsn} we can solve this equation 
for~\eqref{t1t3need}, which can then be written in terms of the
derivatives of $\mathcal{M}_n^{(\beta)}(z,w)$.  Eventually, after some
algebra, we obtain the PDE~\eqref{diffeqshot}. 

\subsection{Non-Perturbative Recurrence Relation for the Joint
  Cumulants}
\label{ss:rr_jcum}

The last result that we need in order to pursue an asymptotic analysis of the 
joint cumulants of conductance and shot noise is the following.
\begin{lemma}
\label{le:jc_rr}
  The cumulants of the joint probability density function of
  conductance and shot noise obey the recurrence relation
 \begin{multline}
\label{shotnoiserec}
k(\eta_{1,4}^{(\beta)}\kappa^{(\beta)}_{l+4,k-1}
-\kappa^{(\beta)}_{l+2,k-1})-2(\alpha+\delta/2+n\beta+2-\beta)
\kappa^{(\beta)}_{l+2,k}+(2l+3k+2)\kappa^{(\beta)}_{l,k+1}\\
+6\eta_{1,4}^{(\beta)}k\sum_{i=0}^{k-1}\binom{k-1}{i}
\sum_{j=0}^{l}\binom{l}{j}\kappa^{(\beta)}_{j+2,i}
\kappa^{(\beta)}_{l-j+2,k-i-1}=\frac{12b^{(\beta)}_{n}}{\beta}k
\mu^{(\beta)}_{l,k-1}
\end{multline}
with boundary conditions $\left
  (\kappa_{l,0}^{(\beta)}\right)_{l=1}^\infty = \left
  (\kappa_{l}^{(\beta)}\right )_{l=1}^\infty$, which are the solutions
of the 
difference equation~\eqref{reccond}.  The coefficients
$\mu_{l,k}^{(\beta)}$ are defined by the recurrence relation
\begin{equation}
\label{jointrecmu}
\mu^{(\beta)}_{l,k} = \sum_{i=0}^{k-1}\binom{k-1}{i}
\sum_{j=0}^{l}\binom{l}{j}r^{(\beta)}_{l-j,k-i}\mu^{(\beta)}_{j,i}
\end{equation}
with boundary conditions $\left (\mu_{l,0}^{(\beta)}\right
)_{l=1}^\infty = \left (\mu_{l}^{(\beta)}\right )_{l=1}^\infty$ 
given by the solution of Eq.~\eqref{recmomcum}. The coefficients
$r_{l,k}^{(\beta)}$ are the ``reduced joint cumulants''
\begin{equation}
\label{jredcumdef}
r^{(\beta)}_{l,k} := \kappa^{(\beta)}_{l,k}\Bigr\rvert_{n \to
  n-i(\beta)}
+\kappa^{(\beta)}_{l,k}\Bigr\rvert_{n \to n+i(\beta)}-2\kappa^{(\beta)}_{l,k}.
\end{equation}
The parameters $\eta_{1,4}^{(\beta)}$, $i(\beta)$ and $b_n^{(\beta)}$
are the same that appeared in Theorem~{\rm \ref{pde:mix_csn}} and
were defined in Eqs.~\eqref{bnbeta} and~\eqref{def_par}.
\end{lemma}
The proof is analogous to that of Lemma~\ref{le:cc_rr} provided the
expansion~\eqref{eq:ts_cond} is replaced by~\eqref{mixcumugen}, the
ODE~\eqref{diffeqcond} by the PDE~\eqref{diffeqshot} and
Eq.~\eqref{eq:rm_exp} by
\begin{equation*}
  \exp\left(\sum_{l=0}^{\infty}\sum_{k=0}^{\infty}
   \frac{r^{(\beta)}_{l,k}z^{l}w^{k}}{l!k!}\right)
 =1 + \sum_{l=1}^{\infty}\sum_{k=0}^{\infty}
\frac{ \mu^{(\beta)}_{l,k}z^{l}w^{k}}{l!k!}. 
\end{equation*}

\begin{remark}
\label{ntc}
Setting $k=0$ in \eqref{shotnoiserec} leads to the 
relation
\begin{equation}
\label{miszero}
(l+1)\kappa^{(\beta)}_{l,1} = (\alpha+\delta/2+n\beta+2-\beta)
\kappa^{(\beta)}_{l+2,0},
\end{equation}
which holds for $\beta \in \{1,2,4\}$. For example, setting $l=0$ and
using \eqref{icreccond_2} gives an exact result for the average shot
noise:
\begin{equation*}
\kappa^{(\beta)}_{0,1} = \frac{1}{4}
\frac{n\bigl(2\alpha+2+\beta(n-1)\bigr)\bigl(\delta+2+\beta(n-1)\bigr)
\bigl(2\alpha+\delta+4+\beta(n-2)\bigr)}{\bigl(\alpha+\delta/2
+2+\beta(n-1)\bigr)\bigl(\alpha+\delta/2+3+\beta(n-1)\bigr)
\bigl(2\alpha+\delta+4+\beta(2n-3)\bigr)}.
\end{equation*}
Instead, setting $l=1$ shows that the covariance between conductance
and shot noise is related in a very simple way to the third cumulant
of the conductance given in Eq.~\eqref{icreccond_3}. When $\beta=2$
and $\alpha=\delta=0$ formula~\eqref{miszero} was derived for all
$l>1$ in \cite{OK09}.  Since when $k=0$ the right-hand side of the
recurrence relation~\eqref{shotnoiserec} does not contribute, we
conjecture that \eqref{miszero} holds not only for $\beta \in
\{1,2,4\}$, but also for any real $\beta>0$. 
\end{remark}

\begin{remark}
  For convenience, in Table~\ref{pie1} we list the first few values of
  the sequence $\left(\mu^{(\beta)}_{l,k}\right)_{l,k=0}^\infty$ 
   appearing on the right hand side
  of (\ref{reccond}). They are expressed in terms of the reduced joint
  cumulants $r^{(\beta)}_{l,k}$, defined in Eq.~\eqref{jredcumdef}. For ease of
notation we have dropped the superscript $(\beta)$ from the formulae in the table. 
Interchanging $l$ and $k$ in $\mu^{{(\beta)}}_{l,k}$ leads to the same expressions
but with the indices of $r^{(\beta)}_{i,j}$ interchanged. 
\end{remark}

\begin{table}[h]
\caption{The first few mixed reduced cumulants}
\centering
    \begin{tabular}{  |c  c | c  |}
    \hline \hline
    l & k & $\mu_{l,k}$ \\ \hline
    $0$ & $0$ & 1 \\ 
    $1$ & $0$ & $r_{1,0}$ \\ 
    $1$ & $1$ & $r_{0,1}r_{1,0}+r_{1,1}$ \\ 
    $2$ & $0$ & $r_{2,0}+r_{1,0}^{2}$ \\ 
    $2$ & $1$ &
    $r_{2,1}+2r_{1,0}r_{1,1}+r_{0,1}r_{2,0}+r_{0,1}r_{1,0}^{2}$\\
    \hline \hline 
    \end{tabular}
\label{pie1}
\end{table}
\begin{remark}
  By setting $(l,k)=(0,1)$ in \eqref{shotnoiserec} we obtain a
  non-perturbative formula for the variance of shot noise. When
  $\delta=0$ this quantity was already calculated
  in~\cite{KSS09,SSW08,KP10a}. Our formulae below are also valid for
  $\delta \neq 0$, and are expressed in terms of
  formulae~\eqref{icreccond} and \eqref{4thcumu}: 
\begin{align*}
\kappa^{(4)}_{0,2} &=
\frac{1}{5}\left((2/3)(\alpha+\delta/2+4n-2)^{2}\kappa^{(4)}_{4}
  -4\kappa^{(4)}_{4}-24\bigl(\kappa^{(4)}_{2}\bigr)^{2}
+\kappa^{(4)}_{2}+3b_{n}^{(4)}\right),\\
\kappa^{(1)}_{0,2} &=
\frac{1}{5}\left((2/3)(\alpha+\delta/2+n+1)^{2}\kappa^{(1)}_{4}
-\kappa^{(1)}_{4}-6\bigl(\kappa^{(1)}_{2}\bigr)^{2}
+\kappa^{(1)}_{2}+12b^{(1)}_{n}\right).
\end{align*}
\end{remark}
\section{Asymptotic Analysis}
\label{as_an}
We shall now prove Theorem~\ref{maintheorem}.  This means
solving the recurrence relations~\eqref{reccond}
and~\eqref{shotnoiserec} at leading order as
$n\to\infty$. It turns out that in this limit, and for each
$\beta \in \{1,2,4\}$, Eqs.~ \eqref{reccond} and \eqref{shotnoiserec} become 
linear homogeneous recurrence relations which can be solved with
elementary methods.

\subsection{Cumulants of the  Conductance}
\label{cumuasympt}
Firstly, we need to study the asymptotic limit as $n \to \infty$ of
the difference equation~\eqref{reccond}, as its solution will provide
the boundary conditions for the leading order contribution to Eq.~\eqref{shotnoiserec}.

\begin{lemma}
\label{le:as_lim}
Let $\beta =1,4$ and suppose that $\alpha$ is independent of $n$,
i.e. $m=n +C$ for some constant $C$, where $m$ and $n$ are the quantum
channels in the left and right leads.  Then, the cumulants of the
conductance admit the asymptotic expansion
\begin{equation}
\label{largencond}
\kappa^{(\beta)}_{l} \sim
n^{-\nu(l)}\sum_{q=0}^{\infty}\kappa^{(\beta)}_{l}(q)n^{-q}, \quad n
\to \infty, \quad l\ge 3,
\end{equation}
where 
\begin{equation}
\label{epsilon}
\nu(l)= l -\epsilon(l) \quad \text{with} \quad 
\epsilon(l) =
\begin{cases}
1 & \text{if  $l$ is even,}\\ 
0 & \text{if $l$ is odd.}
\end{cases}
\end{equation}
Furthermore, the leading order contribution $\kappa^{(\beta)}_l(0)$ satisfies the
difference equation
\begin{multline}
\label{rec1}
-\beta^2(l+1)\kappa^{(\beta)}_{l+1}(0)
+\frac{1}{4}l(l-1)(4l-5)\kappa^{(\beta)}_{l-1}(0)\\
=\frac{3}{16\beta^2}l(l-1)(l-2)\bigl(l-3+\epsilon(l)\bigr)
\bigl(l-4+\epsilon(l)\bigr)\kappa^{(\beta)}_{l-3}(0), \quad l \ge 6,
\end{multline}
with initial conditions
\begin{subequations}
\label{oddbcgenparams}
\begin{align}
\label{oddbcgenparams_1}
\kappa^{(\beta)}_{3}(0) & = \frac{(\delta/2 -\alpha)(\alpha+\delta/2
  +2-\beta)}{4\beta^{4}},\\
\label{oddbcgenparams_2}
 \kappa^{(\beta)}_{5}(0) & = \frac{3(\delta/2 -\alpha)
(\alpha + \delta/2 + 2-\beta)}{4\beta^{6}},
\end{align}
\end{subequations}
for the odd cumulants and
\begin{equation}
\label{eq:even_ic_l}
\kappa^{(\beta)}_{4}(0) =
\left(\frac{\beta}{2}-1\right)\frac{3}{64\beta^{4}}, 
\qquad \kappa^{(\beta)}_{6}(0) = \left(\frac{\beta}{2}-1\right)
\frac{15}{128\beta^{6}}
\end{equation}
for the even cumulants. 
\end{lemma}
Since the parameters $\alpha$ and $\delta$ are independent of $n$, it
follows from the finite-$n$ recurrence relation~\eqref{reccond}
and from the formulae for its coefficients $A(l)$, $B(i,l)$ and
$b_n^{(\beta)}$, see Eqs.~\eqref{eq:pAB} and~\eqref{bnforms}, that the
cumulants are rational functions of $n$; therefore, they admit an
asymptotic expansion of the form~\eqref{largencond}.  The initial
conditions~\eqref{oddbcgenparams} and~\eqref{eq:even_ic_l} are
derived  using the finite-$n$ recurrence relation~\eqref{reccond} and
by computing the leading order term as $n \to \infty$.  They are
consistent with the expansion~\eqref{largencond}.

We shall prove the leading order difference equation~\eqref{rec1} by
appropriately rescaling the recurrence~\eqref{reccond} and using a
dominant balance argument. The explicit expression~\eqref{epsilon} for
the exponent $\nu(l)$ will follow by induction on $l$; that is, we will
assume that $\kappa^{(\beta)}_{l}, \kappa^{(\beta)}_{l-1}, \ldots,
\kappa^{(\beta)}_{3}$ have the asymptotic expansion \eqref{largencond}
and prove that $\kappa^{(\beta)}_{l+1}$ does too for any $l \geq 6$.
We shall discuss the left- and right-hand sides of Eq.~\eqref{reccond}
separately.

\subsubsection{Proof of the left-hand side of Eq.~\eqref{rec1}}
\label{sss:lhs_c}
Recall that 
\begin{align*}
  A(l) & := \eta^{(\beta)}_{1,4}l(l-1)(l-2)+\beta
  l(2l-1) - (6-3\beta)(\alpha+\delta/2+\beta n+2-\beta)l \\
  & \quad \; -(\alpha+\delta/2+\beta n)
  (\alpha+\delta/2+\beta n+2-\beta)(l+1)
\intertext{and}
B(i,l) &:= (2-\chi_{1,2}^{(\beta)})(l-i)\Bigl((l-i)(6i-2)+3\Bigr)
+(\chi_{1,2}^{(\beta)}-1)(l-i)^{2}(6i+2).
\end{align*}
The second term in the left-hand side of Eq.~\eqref{rec1} is 
\begin{multline}
\label{lhslimcond}
\lim_{n \to \infty}n^{l+\epsilon(l)-2}
\Biggl(\eta_{1,4}^{(\beta)}\sum_{i=0}^{l-1}\binom{l}{i}\kappa^{(\beta)}_{i+1}
\kappa^{(\beta)}_{l-i}B(i,l)-l(2l-1)(\alpha-\delta/2 +\beta n)
\kappa^{(\beta)}_{l}\\
-l(l-1)(l-2)\kappa^{(\beta)}_{l-1}\Biggr)
=\frac{1}{4}l(l-1)(4l-5)\kappa^{(\beta)}_{l-1}(0), \quad l\ge 6.
\end{multline} 
The computation of this limit is based on the following
considerations. Because of the asymptotic expansion
(\ref{largencond}), we see that each term of the sum in
Eq.~\eqref{lhslimcond} is of order
$\kappa^{(\beta)}_{i+1}\kappa^{(\beta)}_{l-i} =
O(n^{-l-\epsilon(l)+1})$ and so it does not contribute at leading
order. The exceptions to this come only from terms with
$i\in\{0,1,l-2,l-1\}$, which yield systematic contributions to the 
right-hand side of (\ref{lhslimcond}). By taking only these values of
$i$ in (\ref{lhslimcond}), together with the asymptotic formulae,
\begin{subequations}
\begin{align}
  \label{eq:k1as}
   \kappa_1^{(\beta)} & = \frac{n}{2} + \frac{\alpha -
     \delta/2}{2\beta} + O\left(n^{-1}\right), \qquad n\to
   \infty ,\\
    \label{eq:k2as}
     \kappa_2^{(\beta)}& = \frac{1}{8\beta} + O(n^{-1}), \qquad n \to \infty,
\end{align}
\end{subequations}
one easily derives (\ref{lhslimcond}) after a number of delicate cancellations.

\subsubsection{Proof of the right-hand side of Eq.~\eqref{rec1}}
 
The right-hand side of the difference equation~\eqref{rec1} is
inherited from the right-hand side of the Pfaff-KP
equation~\eqref{pkpequation}, which is the cause of the increased
complexity of the symmetry classes with $\beta=1,4$ compared to those
with $\beta=2$. It is not surprising that it requires more attention
than the left-hand side.

The main problem when trying to solve directly the
ODE~\eqref{diffeqcond} or the finite-$n$ recurrence~\eqref{reccond} is
that both are difference equations in $n$ too.  However, one would
expect that this should not affect the leading order asymptotics of
Eq.~\eqref{reccond} as $n \to \infty$. Thus, the first step is to
compute the leading order term of the reduced moments
$\mu_{l}^{(\beta)}$, which by definition include cumulants of dimensions
$n \pm i(\beta)$.

Recall that the $\mu^{(\beta)}_{l}$'s are expressed in terms of the reduced
cumulants 
\begin{equation}
\label{eq:red_cum_2}
r^{(\beta)}_{l} := \kappa^{(\beta)}_{l}\Bigl\lvert_{n \to n-i(\beta)}
+\kappa^{(\beta)}_{l}\Bigr \rvert_{n \to n+i(\beta)}-2\kappa^{(\beta)}_{l}
\end{equation}
via the recurrence relation~\eqref{recmomcum}.  Since the parameters
$\alpha$ and $\delta$ are independent of $n$, by
assuming~\eqref{largencond} and inserting it into~\eqref{eq:red_cum_2},
we see that  $r^{(\beta)}_l$ admits the asymptotic expansion 
\begin{equation}
\label{reducedexpansion}
r^{(\beta)}_{l} \sim n^{-\nu(l)-2}\sum_{q=0}^{\infty}r^{(\beta)}_{l}(q)n^{-q},
\quad n \to \infty, \quad l \ge 6
\end{equation}
and that the leading order term is
\begin{equation}
\label{reducedlead}
\begin{split}
r^{(\beta)}_{l}(0) & = i(\beta)^2\nu(l)\bigl(\nu(l)+1\bigr)
\kappa^{(\beta)}_{l}(0)  \\
& = i(\beta)^2\bigl(l-\epsilon(l)\bigr)\bigl(l+1-\epsilon(l)\bigr)
\kappa^{(\beta)}_{l}(0).
\end{split}
\end{equation}

In order to relate the asymptotic limit of the reduced moments
$\mu^{(\beta)}_l$ to that of the $r^{(\beta)}_l$'s, it is
convenient to use the formula
\begin{equation}
\label{partitionformula}
\mu^{(\beta)}_{l} = \sum_{\pi}\prod_{B_i \in \pi}r^{(\beta)}_{\abs{B_i}},
\end{equation}
rather than the recurrence relation~\eqref{recmomcum}. In the
right-hand side of~\eqref{partitionformula}, $\pi=\{B_1,\dotsc,B_p\}$,
$1\le p \le l$ runs through all the partitions of a set
$U=\{1,2,\ldots,l\}$ of $l$ elements into disjoint
subsets.  
The notation $\left \lvert B_i \right\rvert$ indicates the cardinality
of the subset $B_i\subseteq U$. Now, the contribution of a partition
$\pi$ to the sum~\eqref{partitionformula} is given by the product
\begin{equation}
\label{bellanalysis}
r^{(\beta)}_{\abs{B_{1}}}\dotsm r^{(\beta)}_{\abs{B_{p}}} =
O\left(n^{\sum_{i=1}^{p}\left[-\abs{B_{i}}-2
+\epsilon\left(\abs{B_{i}}\right)\right]}\right)
=O\bigl(n^{-l-2p + \abs{\pi^{e}}}\bigr), \quad n \to \infty,
\end{equation}
where $\abs{\pi^{\rm e}}$ denotes the number of sets in the partition
$\pi$ with even cardinality. Since $\abs{\pi^{\rm e}}\leq p$, the asymptotic contribution of the product~\eqref{bellanalysis} is 
greatest when $p=1$, which implies $\pi = U$ and  $\abs{\pi^{\rm
    e}}=\epsilon(l)$. Since the number of terms in the
sum~\eqref{partitionformula} remains finite as $n\to \infty$, by
setting $p=1$ and combining Eqs.~\eqref{bellanalysis}
and~\eqref{reducedlead} we obtain 
\begin{equation}
\label{limitofmu}
\lim_{n \to \infty}n^{l+\epsilon(l)-2}\mu^{(\beta)}_{l-3} =
r^{(\beta)}_{l-3}(0)= i(\beta)^2
\bigl(l-3+\epsilon(l)\bigr)\bigl(l-4+\epsilon(l)\bigr)\kappa^{(\beta)}_{l-3}(0). 
\end{equation}

In order to complete the proof of Lemma~\ref{le:as_lim} we 
observe that formula~\eqref{bnforms} gives 
\begin{equation}
\label{limitofb}
b^{(\beta)}_{n}=\frac{1}{256} + O\left(n^{-1}\right), \quad n \to \infty.
\end{equation}
Since when $\beta \in \{1,4\}$, $i(\beta)^2 = 4/\beta$, we arrive at
\begin{equation}
\label{limitingrhs}
  \begin{split}
  \lim_{n\to \infty} n^{l +\epsilon(l)-2} 
\left(12 b_n^{(\beta)}/\beta\right)l(l-1)(l-2)
  \mu_{l-3}^{(\beta)}  & = \frac{3}{16\beta^2}l(l-1)(l-2)\\
   & \quad \times \bigl(l-3 +
  \epsilon(l)\bigr)\bigl(l-4+\epsilon(l)\bigr)\kappa_{l-3}^{(\beta)}.
  \end{split} 
\end{equation}
Combined with the limit \eqref{lhslimcond}, Eq. \eqref{limitingrhs} implies that the asymptotic expansion \eqref{largencond} is valid for every $l\geq3$ by induction. That is, we have
\begin{equation}
  \label{eq:f_l}
  \lim_{n\to \infty} n^{l + \epsilon(l)  - 2}A(l)\kappa^{(\beta)}_{l + 1}
  = -\beta^2(l+1)\kappa_{l+1}^{(\beta)}(0).
\end{equation}
Finally, inserting these limits into \eqref{reccond} yields the recurrence relation \eqref{rec1}.
\subsection{The Solution of the Limiting Recurrence
  Relation~\eqref{rec1}}
\label{ss:sol_rec_cond}
The limiting recurrence relation~\eqref{rec1} is linear and can be
solved using the method of generating functions.  Since either the 
even or the odd cumulants appear in~\eqref{rec1}, we introduce
\begin{subequations}
 \label{eq:cgf_s}
  \begin{align}
   \label{oddgenfun1}
  F^{(\beta)}_{\rm o}(x) & := \sum_{l=1}^{\infty}
   \frac{\kappa^{(\beta)}_{2l+1}(0)}{(2l)!}x^{l},\\
  \label{evengenfun1}
    F^{(\beta)}_{\rm e}(x) & := \sum_{l=2}^{\infty}
    \frac{\kappa^{(\beta)}_{2l}(0)x^{l}}{(2l-1)!}.
  \end{align}
\end{subequations}

\begin{proposition}
  \label{pr:sol_rec_m}
  Let $\beta=1,4$. If at leading order the cumulants of the
  conductance satisfy the recurrence relation~\eqref{rec1} with
  initial conditions~\eqref{oddbcgenparams} and~\eqref{eq:even_ic_l},
  the generating functions~\eqref{eq:cgf_s} are
    \begin{subequations}
      \label{eq:cgf}
      \begin{align}
       \label{eq:ogf}
        F^{(\beta)}_{\rm o}(x)& = \frac{\left(\delta/2
          -\alpha\right)\left(\alpha + \delta/2 + 2 -
          \beta\right)}{2\beta^2}\frac{x/(4\beta^2)}{1- x/(4\beta^2)},\\
        \label{eq:egf}
       F^{(\beta)}_{\rm e}(x) & = \left(\frac{\beta}{2} - 1\right)
        \left(1 - \frac{x}{8\beta^2}
         - \sqrt{1 - x/(4\beta^2)}\right), 
      \end{align}
    \end{subequations}
for $-4\beta^2 < x < 4\beta^2$.
\end{proposition}
\begin{proof}
  In order to simplify the algebra, it is convenient to rescale the
  cumulants $\kappa_l^{(\beta)}$, so that the difference
  equation~\eqref{rec1} becomes independent of the parameters
  $\alpha$, $\beta$ and $\delta$.   Define
  \begin{equation}
    \label{eq:r_cum}
    \tilde{\kappa}_l = \frac{(2\beta)^{l +\epsilon(l) - 1}}%
    {C(\alpha,\delta,\beta)}
      \kappa_{l}^{(\beta)}(0), \qquad l\ge 3,
  \end{equation}
  where $C(\alpha,\delta,\beta)$ is constant in $l$ and is fixed by
  the initial conditions. In terms of the rescaled cumulants
  $\tilde{\kappa}_l$ Eq.~\eqref{rec1} becomes
\begin{multline}
\label{rec1_r}
-(l+1)\tilde{\kappa}_{l+1}
+l(l-1)(4l-5)\tilde{\kappa}_{l-1}\\
= 3l(l-1)(l-2)\bigl(l-3+\epsilon(l)\bigr)
\bigl(l-4+\epsilon(l)\bigr)\tilde{\kappa}_{l-3}, \quad l \ge 6.
\end{multline}
This recurrence relation is linear and independent of $\alpha$,
$\delta$ and $\beta$; therefore, these parameters can affect
the $\tilde{\kappa}_l$'s only through an overall multiplicative factor,
which can be removed by an appropriate choice of
$C(\alpha,\delta,\beta)$.   Then, we introduce
the \textit{rescaled} generating functions as
\begin{subequations}
  \label{eq:res_gf}
\begin{align}
  \label{eq:res_gf_odd}
  \tilde{F}_{\rm o}(y)& :=\frac{1}{C_{\rm o}(\alpha,\delta,\beta)}F^{(\beta)}_{\rm
    o}(4\beta^2 y) = \sum_{l=1}^{\infty}
   \frac{\tilde{\kappa}_{2l+1}(0)}{(2l)!}y^{l}, \\
   \label{eq:res_gf_even}
   \tilde{F}_{\rm e}(y)& :=\frac{1}{C_{\rm e}
    (\alpha,\delta,\beta)}F^{(\beta)}_{\rm
    e}(4\beta^2 y) = \sum_{l=2}^{\infty}
   \frac{\tilde{\kappa}_{2l}(0)}{(2l-1)!}y^{l}.
\end{align}
\end{subequations}

The initial conditions~\eqref{oddbcgenparams} and~\eqref{eq:even_ic_l}
suggest the choice of the  constants $C_{\rm o}$ and $C_{\rm e}$:
\begin{subequations}
\label{eq:C_ic}
\begin{align}
  \label{eq:C_ic_o}
  C_{\rm o}(\alpha,\delta,\beta) & = \frac{\left(\delta/2
          -\alpha\right)\left(\alpha + \delta/2 + 2 -
          \beta\right)}{2\beta^2}, \\
   \label{eq:C_ic_e}
 C_{\rm e}(\alpha,\delta,\beta) & = \left(\frac{\beta}{2} - 1\right).
\end{align}
\end{subequations}
Therefore, the rescaled initial conditions are
\begin{subequations}
  \label{eq:ic_res}
\begin{alignat}{2}
  \label{eq:ic_res_o}
  \tilde{\kappa}_3 & = 2, &\qquad  \qquad & \tilde{\kappa}_5 = 24,\\
     \label{eq:ic_res_e}
  \tilde{\kappa}_4 & = \frac34, && \tilde{\kappa}_{6} =\frac{15}{2}.
\end{alignat}
\end{subequations}

Combining the series expansions in the right-hand sides
of~\eqref{eq:res_gf_odd} and~\eqref{eq:res_gf_even} with the initial
conditions~\eqref{eq:ic_res_o} and~\eqref{eq:ic_res_e}, respectively,
turns the recurrence relation~\eqref{rec1_r} into the remarkably
simple ODEs
\begin{subequations}
  \label{eq:lin_odes}
  \begin{gather}
    \label{eq:lin_ode_o}
    2y(1-3y)(1-y)\frac{d\tilde{F}_{\rm o}(y)}{dy} + (6y^2 -3y +
  1)\tilde{F}_{\rm o}(y) -  3y(1-2y)=0, \\
  4(1-y)\frac{\tilde F_{\rm e}(y)}{dy} + 2\tilde F_{\rm e}(y) - y =0,
  \end{gather}
\end{subequations}
with initial conditions $\tilde{F}_{\rm o}(0)=0$ and $\tilde{F}_{\rm
  e}(0)= 0$, respectively. These equations can be integrated with
elementary methods.  We have
\begin{subequations}
  \label{eq:sim_solas}
  \begin{align}
    \label{eq:simp_sol_o}
     \tilde{F}_{\rm o}(y)& = \frac{y}{1-y},\\
       \label{eq:simp_sol_e}
     \tilde{F}_{\rm e}(y) & = 1 - \frac{y}{2} - \sqrt{1-y}.
  \end{align}
\end{subequations}
We give some additional details on this calculation in Appendix~\ref{se:solrec}.
\end{proof}
\begin{remark}
  The Taylor series of the right-hand sides of Eqs.~\eqref{eq:ogf}
  and~\eqref{eq:egf} are elementary and prove
  Theorem~\ref{maintheorem} when  $l\ge 3$ and $k=0$.
\end{remark}

\subsection{Joint Cumulants of Conductance and 
Shot Noise}
\label{ss:j_sn_c}

In order to complete the proof of Theorem~\ref{maintheorem}, we need
to extend the analysis of the leading order asymptotics of the
cumulants of $G$ to the joint cumulants of $G$ and $P$.   The strategy
follows a similar pattern to the approach used for the leading order
asymptotics of the conductance cumulants, but obviously the technical 
difficulties increase.
\begin{lemma}
  \label{le:lim_rr_jc}
  Let $\beta=1$ or $\beta=4$ and assume that $\alpha$ is independent
  of $n$. 
  Then, for $(l,k) \notin \{(0,0),(1,0),(0,1),(0,2),(2,0)\}$
  the 
  joint cumulants of the conductance and shot noise admit the
  asymptotic expansion
   \begin{equation}
    \label{eq:as_jc}
     \kappa_{l,k}^{(\beta)} \sim n^{-\nu(l,k)}\sum_{q=0}^\infty
     \kappa^{(\beta)}_{l,k}(q)n^{-q}, \quad n \to \infty,
   \end{equation}
  where $\nu(l,k) = l + k -\epsilon(l)$ and $\epsilon(l)$ is defined
  in~\eqref{epsilon}.   Furthermore, the leading order terms
  $\kappa^{(\beta)}_{l,k}(0)$ satisfy the double recurrence relation 
  \begin{multline}
\label{eq:l_rr_jc}
(2l+3k+2)\kappa^{(\beta)}_{l,k+1}(0)
-2\beta\kappa^{(\beta)}_{l+2,k}(0) +(k/2)\kappa^{(\beta)}_{l+2,k-1}(0)\\
=\frac{3k}{16\beta^2}
\bigl(l+k-1-\epsilon(l)\bigr)\bigl(l+k-\epsilon(l)\bigr)
  \kappa^{(\beta)}_{l,k-1}(0),
\end{multline}
with initial conditions given by the sequences
\begin{subequations}
    \label{eq:ic_as} \begin{align}
    \label{eq:ic_as_o}
    \kappa^{(\beta)}_{l,0}(0) & = (\delta/2 -\alpha)(\alpha + \delta/2 + 2 -
    \beta)\frac{(l-1)!}{\beta(2\beta)^l}, &&\text{when $l\ge 3$
      and odd,}\\
\intertext{and}
    \label{eq:ic_as_e}
\kappa_{l,0}^{(\beta)} (0)&
=\left(\frac{\beta}{2}-1\right)\frac{(l-2)!}{\left(4\beta\right)^l}
\binom{l}{l/2},  &&  
\text{when $l\ge 4$ and even.}
  \end{align}
\end{subequations}
\end{lemma}

The initial 
conditions~\eqref{eq:ic_as_o} and~\eqref{eq:ic_as_e} are obtained by
expanding the generating function in
Proposition~\ref{pr:sol_rec_m}. In the proof of Lemma~\ref{le:as_lim}, Sec.~\ref{cumuasympt}, it was
argued that the initial conditions $\kappa_{l,0}^{(\beta)}$ are
rational functions of $n$.  It immediately follows from
Eq.~\eqref{shotnoiserec} that the joint cumulants
$\kappa_{l,k}^{(\beta)}$ are rational functions of $n$ too and therefore have an asymptotic expansion in powers of $n^{-1}$ as $n \to \infty$. 

The precise form of the expansion     \eqref{eq:as_jc}, as well as the value of the exponent $\nu(l,k)$, will be determined by mathematical induction on $k$. First one must deal with a number of base cases. Note that if $k=0$, Eq. \eqref{eq:as_jc} reduces to our earlier result \eqref{largencond}; if $k=1,2$, Eq. \eqref{eq:as_jc} follows easily from the recurrence \eqref{shotnoiserec} provided $l\neq 0$. If $k=3$, we again use \eqref{shotnoiserec} to obtain \eqref{eq:as_jc}, this time for any $l\geq 0$. In the following we shall proceed by induction on $k$ to prove \eqref{eq:as_jc} for any $l\geq 0$ and $k>3$.
We discuss the left- and right-hand sides of Eq.~\eqref{eq:l_rr_jc}
separately.  

\subsubsection{Proof of the left-hand side of Eq.~ \eqref{eq:l_rr_jc}}
\label{sss:lhs_jc}

By rescaling and taking the limit of the left-hand side of
Eq.~\eqref{shotnoiserec} we obtain
\begin{multline}
\label{shotasymptlhs}
\lim_{n \to \infty}n^{l+k+1-\epsilon(l)}\left(
  k(\eta_{1,4}^{(\beta)}\kappa^{(\beta)}_{l+4,k-1}
  -\kappa^{(\beta)}_{l+2,k-1})-2(\alpha+\delta/2+n\beta+2-\beta)
  \kappa^{(\beta)}_{l+2,k}\right.\\
\left. +(2l+3k+2)\kappa^{(\beta)}_{l,k+1} + 6\eta_{1,4}^{(\beta)}
  k\sum_{i=0}^{k-1}\binom{k-1}{i}\sum_{j=0}^{l}
  \binom{l}{j}\kappa^{(\beta)}_{j+2,i}\kappa^{(\beta)}_{l-j+2,k-i-1}\right)\\
= (2l+3k+2)\kappa^{(\beta)}_{l,k+1}(0)
-2\beta\kappa^{(\beta)}_{l+2,k}(0)
+ (k/2)\kappa^{(\beta)}_{l+2,k-1}(0).
\end{multline}
In order to prove this limit, one first has to verify that it holds
for the initial conditions~\eqref{eq:ic_as} by direct
substitution. Then, by assuming the expansion~\eqref{eq:as_jc}, we see
that the only leading order contributions from the double sum in
\eqref{shotasymptlhs} come from the choice of indices $(j,i)=(0,0)$
and $(j,i)=(l,k-1)$. Elementary manipulations involving the limiting
value of the variance of the conductance give the right-hand side of
\eqref{shotasymptlhs}.

\subsubsection{Proof of the right-hand side of Eq.~\eqref{eq:l_rr_jc}}
\label{sss:rhs_jc}

As in the proof of Lemma~\ref{le:as_lim}, deriving the right-hand side of the
recurrence relation~\eqref{eq:l_rr_jc} presents more technical
difficulties than the left-hand side. The reason is that the finite-$n$ recurrence \eqref{shotnoiserec} involves the coefficients $\mu^{(\beta)}_{l,k}$, defined in terms of the auxiliary recurrence \eqref{jointrecmu} involving the two dimensional sequence of \textit{reduced} cumulants 
\begin{equation}
\label{redgendefb1}
r^{(\beta)}_{l,k} := \kappa^{(\beta)}_{l,k}\Bigr \rvert_{n \to n-i(\beta)}
+\kappa^{(\beta)}_{l,k}\Bigr\rvert_{n \to n+i(\beta)}
-2\kappa^{(\beta)}_{l,k}.
\end{equation}

Since we assume that the parameters $\alpha$ and $\delta$ are
independent of $n$, we can insert the expansion~\eqref{eq:as_jc} into 
 \eqref{redgendefb1} to obtain

\begin{equation}
r^{(\beta)}_{l,k} \sim
n^{-\nu(l,k)-2}\sum_{q=0}^{\infty}r^{(\beta)}_{l,k}(q)n^{-q}, \quad
n\to \infty.
\end{equation}
The coefficients in this series may be computed in terms of those in
the expansion~\eqref{eq:as_jc}. At leading order in $n$, we find 
\begin{equation*}
\begin{split}
r^{(\beta)}_{l,k}(0) &= i(\beta)^2\nu(l,k)\bigl(\nu(l,k)+1\bigr)
\kappa^{(\beta)}_{l,k}(0)\\ 
& = i(\beta)^2\bigl(l+k-\epsilon(l)\bigr)\bigl(l+k+1
-\epsilon(l)\bigr)\kappa^{(\beta)}_{l,k}(0).
\end{split}
\end{equation*}

Let us write the coefficient
$\mu^{(\beta)}_{l,k}$ as a sum over partitions,
\begin{equation}
\label{partitionformula2}
\mu^{(\beta)}_{l,k} = \sum_{\pi}\prod_{B_i \in \pi}
r^{(\beta)}_{\abs{L(B_i)},\abs{M(B_i)}}.
\end{equation}
The sum varies over all partitions $\pi =\{B_1,\dotsc,B_p\}$, $1 \le p
\le l + k$, of a set $U = \{1,\ldots,l+k\}$ into disjoint subsets
$B_i$. For each partition $\pi$, and for each subset $B_i$, we further
divide $B_i$ into the two subsets $L(B_i) := \{j \in B_i: j \leq l\}$ and
$M(B_i) := B_i\setminus L(B_i)$. We argue that there is only one 
partition $\pi^{*}$ which gives the desired asymptotic contribution
as $n \to \infty$.

The contribution of a partition $\pi=\{B_{1},\ldots,B_{p}\}$ to the
sum~\eqref{partitionformula2} is given by the product
\begin{equation}
\label{pf2analysis}
\begin{split}
r^{(\beta)}_{\abs{L(B_{1})},\abs{M(B_{1})}}\times\dotsm \times 
r^{(\beta)}_{\abs{L(B_{p})},\abs{M(B_{p})}}&
=O\left(n^{-\sum_{i=1}^{p}\left[\abs{L(B_{i})}+\abs{M(B_{i})}-\epsilon(L(B_{i}))
  +2\right]}\right)\\
&=O\left(n^{-\sum_{i=1}^{p}\left[\abs{B_{i}}-\epsilon(L(B_{i}))+2\right]}\right)\\
&=O\left(n^{-l-k-2p+\abs{\pi^{\rm e}_{L}}}\right),
\end{split}
\end{equation}
where $\abs{\pi^{\rm e}_{L}}$ denotes the number of subsets $B_i$ in
the partition $\pi$ such that $L(B_i)$ has even cardinality. As with the conductance, 
the dominant contribution to \eqref{pf2analysis} corresponds to $p=1$,
\textit{i.e.} when the partition $\pi=\pi^{*} = U$. In this case, we
have $L(U)=\{1,\ldots,l\}$ and $M(U) = \{l+1,\ldots,l+k\}$, and so 
$\abs{\pi^{\rm e}_{L}}=\epsilon(l)$. Therefore, we find that
\begin{equation}
\label{jointmulim}
\lim_{n \to \infty}n^{l+k+1-\epsilon(l)}\mu^{(\beta)}_{l,k-1} =
i(\beta)^2
\bigl(l+k-1-\epsilon(l)\bigr)\bigl(l+k-\epsilon(l)\bigr)
\kappa^{(\beta)}_{l,k-1}(0).
\end{equation}
Using the leading order of $b_n^{(\beta)}$~\eqref{limitofb} and the
equality $i(\beta)^2 = 4/\beta$ we arrive at
\begin{equation}
  \label{eq:fp_jc_rr}
  \lim_{n\to \infty}n^{l + k +1-\epsilon(l)}\left(\frac{12 
      b^{(\beta)}_n}{\beta}k \mu^{(\beta)}_{l,k-1}\right) =
  \frac{3}{16\beta^2} \bigl(l+k-1-\epsilon(l)\bigr)\bigl(l+k-\epsilon(l)\bigr)
\kappa^{(\beta)}_{l,k-1}(0).
\end{equation}
Combining Eq.~\eqref{eq:fp_jc_rr} with the limit \eqref{shotasymptlhs}
implies that the asymptotic expansion~\eqref{eq:as_jc} is valid for
the stated values of $l$ and $k$. Inserting these limits into the
finite-$n$ recurrence \eqref{shotnoiserec} finally completes the proof
of Lemma~\ref{le:lim_rr_jc}.

\subsection{The Solution of the Limiting Recurrence Relation}
\label{sse:l_s_n}

We now complete the proof of Theorem~\ref{maintheorem} by completing
the derivation of Eqs.~\eqref{oddmixed} and~\eqref{evenmixed}.

Firstly, we remove the dependence on $\alpha$, $\beta$ and $\delta$
from the recurrence relation~\eqref{eq:l_rr_jc} and its initial
conditions~\eqref{eq:ic_as} by rescaling the joint cumulants
$\kappa_{l,k}^{(\beta)}$.  We write
\begin{equation}
  \label{eq:rs_jc}
  \tilde{\kappa}_{l,k} = \frac{\left(4\beta\right)^k
    \left(2\beta\right)^{l + \epsilon(l) -1}}{C(\alpha,\beta,\delta)}
   \kappa_{l,k}^{(\beta)}(0).
\end{equation}
As discussed in Sec.~\ref{ss:sol_rec_cond}, since
Eq.~\eqref{eq:l_rr_jc} is linear, it is not affected by the constants
$C(\alpha,\beta,\delta)$, which are fixed by the initial
conditions~\eqref{eq:ic_as_o} and~\eqref{eq:ic_as_e}; therefore, for
convenience, we choose the same as in Eqs.~\eqref{eq:C_ic_o}
and~\eqref{eq:C_ic_e}.  Substituting~\eqref{eq:rs_jc}
into~\eqref{eq:l_rr_jc} yields
\begin{multline}
\label{solveme}
(2l+3k+2)\tilde{\kappa}_{l,k+1}+2k\tilde{\kappa}_{l+2,k-1}
-2\tilde{\kappa}_{l+2,k}\\
=3k\bigl(l+k-1-\epsilon(l)\bigr)\bigl(l+k-\epsilon(l)\bigr)
\tilde{\kappa}_{l,k-1}.
\end{multline}

We treat separately the cases $l$ even and $l$
odd. Replacing $l$ with $2l$ gives
\begin{equation}
\label{lrse}
(4l+3k+2)\tilde{\kappa}_{2l,k+1}+2k\tilde{\kappa}_{2l+2,k-1}
-2\tilde{\kappa}_{2l+2,k}=3k\bigl(2l+k-2\bigr)\bigl(2l+k-1\bigr)
\tilde{\kappa}_{2l,k-1}.
\end{equation}
We need to solve this difference equation subject to the initial
conditions
\begin{equation}
\label{conic1}
\tilde \kappa_{2l,0} = \frac{(2l-2)!}{4^l}\binom{2l}{l}.
\end{equation}
We claim that the unique solution is given by
\begin{equation}
\label{lrsesol}
\tilde \kappa_{2l,k} = (-1)^k \frac{(2l+k-2)!}{4^{l}}
\sum_{j=0}^{k}\binom{2j+2l}{j+l}\binom{k}{j}(-2)^{-j}.
\end{equation}
This formula is proved in Appendix~\ref{ap:p_th} using the properties
of Jacobi polynomials.

For the odd cumulants, replacing $l$ with $2l-1$ in~\eqref{solveme} leads to
\begin{equation}
\label{lrso}
(4l+3k)\tilde{\kappa}_{2l-1,k+1}+2k\tilde{\kappa}_{2l+1,k-1}
-2\tilde{\kappa}_{2l+1,k}=3k\bigl(2l+k-2\bigr)\bigl(2l+k-1\bigr)
\tilde{\kappa}_{2l-1,k-1}.
\end{equation}
We need to solve this recurrence subject to the initial condition
\begin{equation}
\tilde \kappa_{2l-1,0}= (2l-2)!
\end{equation}
It is straightforward to prove by direct substitution that the unique 
solution is 
\begin{equation}
\tilde \kappa_{2l-1,k}= (k+2l-2)!.
\end{equation}

This completes the proof of Theorem~\ref{maintheorem}.

\subsection{Higher Order Asymptotics}
\label{higherorder}

As we discussed in Remark~\ref{staircase}, the cumulants of the shot
noise \eqref{poddconj} of odd index appear at sub-leading order in the
large-$n$ expansion \eqref{eq:as_jc}. Therefore, in order to
prove~\eqref{poddconj}, we need to calculate $\kappa^{(1)}_{0,k}(1)$
for $k$ odd.  The purpose of this section is to give a flavour of the
difficulties one encounters when trying to compute next to leading
order terms of the cumulants using our approach.

In order to simplify the analysis, we consider the case with
$\beta=1$, $\alpha=-1/2$ and, $\delta=0$, \textit{i.e.}  the
scattering matrix belongs to the COE and both leads carry the same
number of channels.  

Inserting the asymptotic expansion~\eqref{eq:as_jc} into the
recurrence relation~\eqref{shotnoiserec}, setting $l=0$ and equating
terms of order $n^{-k}$ for $k$ odd  yields
\begin{multline*}
  (k/2)\kappa^{(1)}_{2,k-1}(1)+(3k+2)\kappa^{(1)}_{0,k+1}(1)
-2\kappa^{(1)}_{2,k}(1)-(3/16)k^{2}(k-1)\kappa^{(1)}_{0,k-1}(1)\\
+6k\sum_{i=1}^{k-2}\binom{k-1}{i}\kappa^{(1)}_{2,k-i-1}(0)
\kappa^{(1)}_{2,i}(0)-(3k/4)\kappa^{(1)}_{2,k-1}(0)-\kappa^{(1)}_{2,k}(0).
\end{multline*}
From this equation we can see that the cumulants
$\kappa^{(1)}_{0,k}(1)$ are expressed in terms of the mixed cumulants
$\kappa^{(1)}_{2,k-1}(1)$ and $\kappa^{(1)}_{2,k}(1)$. Thus, it becomes
clear that to compute the $\kappa^{(1)}_{0,k}(1)$'s one requires the
knowledge of $\kappa^{(1)}_{2l,k}(1)$ for each $l$. 

Define $g_{l,k}:=\kappa^{(1)}_{2l,k}(1)$.  Then, at next to leading
order Eq.~\eqref{shotnoiserec} becomes
\begin{equation*}
(k/2)g_{l+1,k-1}+(4l+3k+2)g_{l,k+1}-2g_{l+1,k}-
(3k/16)(2l+k-1)(2l+k)g_{l,k-1}=\Theta(l,k), 
\end{equation*}
where $\Theta(l,k)$ is a complicated function of the known sequence
$\kappa^{(1)}_{2i,j}(0)$. 

Using \texttt{MAPLE} we were able to surmise that for $l>1$ the
quantities
\begin{equation*}
h_{l,k} = \frac{\kappa^{(1)}_{2l,k}(1)}{(2l+k-1)!}
-\frac{4\kappa^{(1)}_{2l+2,k}(1)}{(2l+k+1)!}
\end{equation*}
can be extracted from the generating function
\begin{equation*}
h_{l,k} = \res_{x=0}\left(x^{-l-k-2}\sqrt{4-x}\res_{y=0}\left(y^{1-k} 
F(x,y)\right)\right),
\end{equation*}
where 
\begin{equation}
\label{fxy}
F(x,y) = \frac{\left(\int_{}^{y}
\frac{\sqrt{y'}P(x,y')}{4(4-y'x)\sqrt{y'x-8y'+4}}dy'\right)}
{\bigl((yx-8y+4)y\bigr)^{3/2}}.
\end{equation}
The function $P(x,y)$ is the polynomial
\begin{equation*}
P(x,y)=-512+256x-20x^{2}+64y^{2}x^{2}-128yx+y^{2}x^{4}-16y^{2}x^{3}
+3yx^{3}-8yx^{2}
\end{equation*}

We remark that the integral in Eq.~\eqref{fxy} can be computed in
closed form, but it is rather long and cumbersome to be reported
here.  

\section{The Wigner Delay Time}
\label{se:wdt}
The analysis of the random variable~\eqref{eq:w_dt2} follows the same
strategy that we adopted to study the conductance.  The main
difference is that the exponent~\eqref{eq:par_lag} depends on $n$
explicitly, which means that the analysis of the asymptotic behaviour
of the cumulants is more complicated.  A direct study of the
integral~\eqref{wignermgf} would require the solution of a delicate
double scaling problem. In addition, the linear homogeneous recursion relations obtained in the previous section 
must be replaced with non-linear recursions that are harder to
solve. Another difficulty is the \textit{singular} nature of the
weighting function~\eqref{eq:wtd}, which implies that the cumulant
generating function is not analytic and only a finite number of
moments exist. 

\subsection{The Differential-Difference Equation}

The proof of Theorem~\ref{th:wdt_ode} is entirely analogous to that of
Theorem~\ref{ode:con}.  Therefore, we shall only outline the relevant
steps. 

Recall that the $\tau$-function associated to the
Wigner delay time is
\begin{equation}
\label{eq:tau_wdt}
\tau_{n}^{(\beta)}(\mathbf{t};z) := \frac{1}{n!}\int_{\mathbb{R}_+^{n}}
 \prod_{j=1}^{n}\rho_z(\lambda_{j})\exp\left(
 \sum_{i=1}^{\infty}t_{i}\lambda_{j}^{i}\right)
 \prod_{1\le j < k \le n}\abs{\lambda_k -
   \lambda_j}^{\beta}d\lambda_1\dotsm d\lambda_n,
\end{equation}
where $\rho_z(\lambda)$ was defined in Eq.~\eqref{eq:wtd}.  It is
related to the moment generating function by the equation
\begin{equation*}
  M_n^{(\beta)}(\beta y/2) =
  \frac{n!}{m^{(\beta)}_n}\tau(\mathbf{t};\beta y/2)\biggr
  \rvert_{\mathbf{t}= \mathbf{0}}.
\end{equation*}
Because of its universal usage and since there is no risk of
ambiguity, we adopt the same notation to indicate the $\tau$-function
that we used for the conductance and for the shot noise.

This $\tau$-function satisfies the KP equation \eqref{kpequation} or
Pfaff-KP equation \eqref{pkpequation} depending on the value of
$\beta$. The definition~\eqref{eq:tau_wdt} gives
\begin{equation}
\label{eq:zt1d_wdt}
\frac{\partial}{\partial t_{1}}\tau^{(\beta)}_{n}(\mathbf{t};z) 
= \frac{\partial}{\partial z}\tau^{(\beta)}_{n}(\mathbf{t};z).
\end{equation}
Thus, the derivative
\begin{equation*}
  \xi^{(\beta)\prime}_n(z) = \frac{d \log M^{(\beta)}_n(z)}{dz}
\end{equation*}
can be expressed in terms of the projection at $\mathbf{t}=
\mathbf{0}$ of the right-hand side of Eq.~\eqref{eq:zt1d_wdt}.  It
remains to express the partial derivatives
\begin{equation}
\label{missingpds_wdt}
\frac{\partial^{2}}{\partial
  t_{2}^{2}}\log\tau^{(\beta)}_{n}(\mathbf{t};z)
\biggr \rvert_{\mathbf{t}=\mathbf{0}}\quad \text{and} \quad 
 \frac{\partial^{2}}{\partial t_{1}\partial t_{3}}
\log\tau^{(\beta)}_{n}(\mathbf{t};z)\biggr \rvert_{\mathbf{t}=\mathbf{0}}
\end{equation}
in terms of the derivative of $\xi^{(\beta)}_n(z)$.  The extra relations needed for 
this purpose are provided by the Virasoro constraints.
 
The only Virasoro operators needed for these calculations are 
$\mathcal{V}_{-1}$ and $\mathcal{V}_0$.  They are expressed in terms
of the parameters $(a_i)_{i=0}^\infty$  and $(b_i)_{i=0}^\infty$ 
defined in Eq.~\eqref{eq:b_c}, which are derived from the
relation
\begin{equation*}
 \frac{d \log \rho_z(\lambda)}{d\lambda} 
= \frac{z\lambda^2 - b\lambda   +1}{\lambda^2}.
\end{equation*}
This equation gives
\begin{equation*}
a_{2}=1, \quad b_{0}=-1, \quad b_{1}=b,\quad b_{2}=-z,
\end{equation*}
and otherwise $a_{i}=0$ and  $b_{i}=0$. Finally, a lengthy but
straightforward calculation analogous to that in
Appendix~\ref{virconapp} leads to the
ODE~\eqref{wigode}. 

\subsection{Finite $n$ Recursion Relation}
\label{ss:fn_wdt}

Let us introduce the \textit{reduced cumulants}
\begin{equation}
\label{eq:red_cum_wdt}
\rho^{(\beta)}_{l} := K^{(\beta)}_{l}\Bigr \rvert_{n \to n-i(\beta)}
+K^{(\beta)}_{l}\Bigr \rvert_{n \to n+i(\beta)}-2K^{(\beta)}_{l},
\end{equation}
where the $K_l^{(\beta)}$'s are the coefficients of the
expansion~\eqref{eq:cum_ex_wdt}, as well as the \textit{reduced moments } 
\begin{equation*}
\phi^{(\beta)}_{l}=\sum_{j=0}^{l-1}\binom{l-1}{j}
\rho^{(\beta)}_{l-j}\phi^{(\beta)}_{j}, \qquad \phi^{(\beta)}_{0}=1.
\end{equation*}
We also define 
\begin{subequations}
\label{eq:c_rr_wdt}
\begin{align}
\label{eq:c_rr_wdt_a}
A(l) & := \eta_{1,4}^{(\beta)}l(l-1)(l-2)+\beta l(2l-1) 
-\omega(\omega-4+2\beta)l
-\omega(\omega+2-\beta), \\
\label{eq:c_rr_wdt_b}
B(i,l) & := (l-i)\Bigl(6\eta_{1,4}^{(\beta)}(l-i-1)i+\beta+4\beta i\Bigr), 
\end{align}
\end{subequations}
where $\omega := b-2-\beta n +\beta$.  Recall also that $q :=
\left \lfloor b-2 - \beta(n-1)\right \rfloor$. 

\begin{lemma}
  \label{le:rr_wtd}
  The cumulants of the Wigner delay time are the unique solution of
  the recurrence relation
\begin{multline}
\label{wigrec}
A(l)K^{(\beta)}_{l+1}+\sum_{i=0}^{l-1}\binom{l}{i}B(i,l)K^{(\beta)}_{l-i}
K^{(\beta)}_{i+1}+\beta l(2l-1)K^{(\beta)}_{l}\\
=(12d^{(\beta)}_{n}/\beta)l(l-1)(l-2)(\beta/2)^{4}\phi^{(\beta)}_{l-3}, \quad 3 
\le l \le q-1,
\end{multline}
with initial conditions given by Eqs.~\eqref{wigcumus}.   
\end{lemma}

The proof is entirely analogous to that of Lemma~\ref{le:cc_rr}.
Inserting the expansion~\eqref{eq:cum_ex_wdt} into~\eqref{wigode}
and equating the first three powers of $z$ gives a system of linear
equations whose solutions are the initial conditions~\eqref{wigcumus}.
Equating higher powers leads to~\eqref{wigrec}. 

\begin{remark}
  Equation~\eqref{eq:red_cum_wdt} depends implicitly on the exponent
  $b$ that appears in $\rho_z(\lambda)$.  Since the Pfaff-KP
  equation~\eqref{pkpequation} was derived with the assumption that
  the weighting function in the integral~\eqref{taugencond} should be
  the same for any dimension $n$ of the $\tau$-function, the reduced
  cumulants are well defined only if all the $K_l^{(\beta)}$'s in the
  right-hand side are computed for fixed $b$.  Once a solution
  $\bigl(K_l^{(\beta)}\bigr)_{l=1}^q$ of the difference
  equation~\eqref{wigrec} is found, the exponent $b$ can be assigned
  the value~\eqref{eq:par_lag} to recover the cumulants.
\end{remark}

The implicit dependence of the reduced
cumulants~\eqref{eq:red_cum_wdt} on the exponent~\eqref{eq:par_lag}
means that the asymptotic analysis of the cumulants of
$\tau_{\mathrm{W}}$ for $\beta =1,4$ involves a double scaling limit;
so does a direct study of the asymptotic behaviour of integral
representation~\eqref{wignermgf} of $M_n(z)$. These problems go beyond
the scope of this paper and remain open.  In the next section we shall
study the asymptotic limit of the cumulants for $\beta=2$.  In this case the right-hand side of the recurrence
relation~\eqref{wigrec} is zero and a direct study of the limit $n \to
\infty$ presents less challenges.

\subsection{The Limiting Recurrence Relation: $\beta=2$} 
\label{ss:lim_rec_wdt}

The proof of Theorem~\ref{th:lim_cum_wdt} combines the approaches in
Lemma~\ref{le:as_lim} and Proposition~\ref{pr:sol_rec_m}, which provide
the weak localization corrections for the cumulants of the
conductance.  The main difference is that the limiting recurrence
relation is not linear.  However, it can still be solved using the
method of generating functions. 

Equations~\eqref{eq:c_rr_wdt} and~\eqref{wigrec} imply that the 
cumulants are a rational function of $n$.  A direct calculation 
shows that $K^{(\beta)}_l$ admits the asymptotic expansion
\begin{equation}
\label{wigasy}
K^{(\beta)}_{l} \sim
n^{2-2l}\sum_{p=0}^{\infty}K^{(\beta)}_{l}(p)n^{-p}, 
\quad n \to \infty, \quad l>0.
\end{equation}
This formula is correct for all the symmetry classes $\beta \in
\left\{1,2,4\right\}$.  Note that the exponent in front of the sum
does not contain a term analogous to $\nu(l)$; therefore, we do not see 
the ``staircase'' asymptotic behaviour which characterized the
cumulants of conductance and shot noise (see
Remark~\ref{staircase}). 

When $\beta=2$, our finite-$n$ recurrence relation (\ref{wigrec}) simplifies considerably:
\begin{equation}
\label{beta2wigfn}
(l+1)(n^{2}-l^{2})K^{(2)}_{l+1}=2l(2l-1)K^{(2)}_{l}+2\sum_{i=0}^{l-1}(3i+1)\binom{l}{i}(l-i)^{2}K^{(2)}_{i+1}K^{(2)}_{l-i},	
\end{equation}
where we made use of the identity
\begin{equation}
\sum_{i=0}^{l-1}\binom{l}{i}(l-i)(6(l-i-1)i+2+8i)x_{i+1}x_{l-i} = 2\sum_{i=0}^{l-1}(3i+1)\binom{l}{i}(l-i)^{2}x_{i+1}x_{l-i},
\end{equation}
which follows from equating coefficients of the indeterminate
$x_{i+1}x_{l-i}$. Inserting the asymptotic expansion (\ref{wigasy})
into (\ref{beta2wigfn}) and taking the leading order terms 
gives the limiting recurrence~\eqref{limwigrecintro}, after
identifying the new variable
\begin{equation}
p_{l}=\frac{K^{(2)}_{l}(0)}{(l-1)!}.
\end{equation}

Difference equations similar to~\eqref{limwigrecintro} have been
studied in~\cite{HJK02,JK01} using formal power series solutions of
Painlev\'e I. In a similar spirit Garoufalidis \textit{et
  al.}~\cite{GIKM12} used the Riemann-Hilbert approach to study
asymptotics of non-linear recurrence relations.

Recall that 
\begin{equation}
\label{eq:gf_c2}  
F(z) = \sum_{i=1}^\infty p_iz^i.
\end{equation}
This generating function $F(z)$ turns the difference
equation~\eqref{limwigrecintro} into the non-linear ODE 
\begin{equation}
\label{limwigode}
2F(z)+F'(z)-4zF'(z)-6zF'(z)^{2}+4F(z)F'(z)=1,
\end{equation}
with initial condition $F(0)=0$.  The constant in the right-hand side
is fixed by the condition that the first cumulant is normalized to
one.  Equation~\eqref{limwigode} is of d'Alembert type and can be 
integrated.  We have
\begin{equation}
\label{solwigode2}
F(z) = 3/2-(3/2)\Omega(z)+2z\Omega(z)^{3}+3z\Omega(z)^{2}+2z\Omega(z), 
\end{equation}
where $\Omega(z)$ is the unique solution of the quartic equation.
\begin{equation}
\label{quartic}
4z\Omega^{4}+8z\Omega^{3}+(4z-3)\Omega^{2}+2\Omega+1=0,
\end{equation}
with $\Omega(0)=1$. 

Solving (\ref{quartic}) for $z$ in terms of $\Omega$ gives 
\begin{equation}
z(\Omega) = \frac{1}{4}\frac{(3\Omega+1)(\Omega-1)}
{\Omega^{2}(\Omega+1)^{2}}. 
\end{equation}
Since, as a function of $\Omega$, $z(\Omega)$ is analytic at $1$, $z(1)=0$ and $z'(1) \neq 0$, we can 
apply the Lagrange inversion formula and obtain the formal power
series
\begin{equation}
\label{eq:om_ex}
\Omega(z) = 1+\sum_{k=1}^{\infty}\lim_{w \to 1}
\frac{d^{k-1}}{dw^{k-1}}\left(\frac{w-1}{z(w)}\right)^{k}\frac{z^{k}}{k!} 
= 1+\sum_{k=1}^{\infty}\zeta_{k}z^{k},
\end{equation}
where 
\begin{equation}
\label{intwig}
\begin{split}
\zeta_{k} &= \lim_{w \to 1}\frac{d^{k-1}}{dw^{k-1}}
\left(\frac{4w^{2}(w+1)^{2}}{3w+1}\right)^{k}\frac{1}{k!}\\
&=\frac{4}{k}\sum_{i=0}^{k-1}
\binom{2k-i-2}{k-1}\left(\sum_{p=0}^{i}
\binom{2k}{p}\binom{2k}{i-p}2^{i+p}\right)(-3)^{k-1-i}.
\end{split}
\end{equation}
The right-hand side follows from computing the derivatives and
rewriting the resulting factorials in terms of binomial
coefficients. 

In order to complete the proof of Theorem~\eqref{th:lim_cum_wdt}, we
are left to show that the Taylor coefficients of the generating
function~\eqref{eq:gf_c2}, which are the solution
to~\eqref{limwigrecintro}, are integers. It follows from
Eqs.~\eqref{solwigode2} and~\eqref{eq:om_ex} that they are integers if
the coefficients $\zeta_k$ are integers too. This is a straightforward
consequence of the following.
\begin{lemma}
Let $k$, $i$ and $p$ be integers such that $k >0$, $i=0,\dotsc,k-1$ and
$p=0,\dotsc,i$. The product
\begin{equation}
  \label{eq:prod_bin}
  \binom{2k}{p}\binom{2k}{i-p}\binom{2k-i -2}{k-1}
\end{equation}
is divisible by $k$.  
\end{lemma}
\begin{proof}
Since $\frac{1}{k}\binom{2k-2}{k-1}$ is a Catalan number and 
\[
\binom{2k}{1}= 2k, \qquad \binom{2k}{2} = k(2k-1),
\] 
we can assume that $k \ge 4$, $i \ge 3$ and $p \ge3$.

Denote by $\gcd(a_1,\dotsc,a_m)$ and $\lcm(a,\dotsc,a_m)$ the greatest common
divisor and least common multiple of the integers $a_1,\dotsc,a_m$. For any
$a > b > 1$ the identity
$$
\binom{a}{b} = \frac{a}{b}\binom{a-1}{b-1}
$$
implies that $ a/\gcd(a,b)$ and $b/\gcd(a,b)$ divide
$\binom{a}{b}$ and $\binom{a-1}{b-1}$, respectively. 
Therefore, 
\[
k_1 := \frac{k}{\gcd(k,2k- i - 1)}
\]
is a factor of $\binom{2k-i - 2}{k-1}$; similarly, $k/\gcd(k,p)$
and $k/\gcd(k,i-p)$
divide $\binom{2k}{p}$ and $\binom{2k}{i-p}$, respectively.\footnote{When $i=p$,
  we set by convention $\gcd(k,0)=k$.} 

If a common divisor of $k$ and $i-p$ is a factor of either $i$ or $p$, then it
must divide both $i$ and $p$.  It follows that
\begin{align*}
 \mathfrak{g} & :=  \gcd\bigl(\gcd(k,p),\gcd(k,i-p)\bigr) = \gcd(k,i,p) \\
  \mathfrak{m} & := \lcm\bigl(\gcd(k,p),\gcd(k,i-p)\bigr) =
  \frac{\gcd(k,p)\gcd(k,i-p)}{\gcd(k,i,p)}.
\end{align*}
Therefore, $\binom{2k}{p}\binom{2k}{i-p}$ is divisible by $k^2/(\mathfrak{gm})$.

Consider the prime factorizations 
\[
\gcd(k,p) = \prod_\rho q_\rho^{\alpha_\rho} \quad
\text{and}  \quad \gcd(k,i-p) = \prod_\rho q_\rho^{\beta_\rho}.
\]
By definition $\mathfrak{m} = \prod_\rho q_\rho^{\gamma_\rho}$, where $\gamma_\rho :=
\max\{\alpha_\rho,\beta_\rho \}$.  Since for each $\rho$, $q_\rho^{\gamma_\rho}
\mid k$, then $\mathfrak{m} \mid k$ too.  As a consequence,
$k^2/(\mathfrak{gm})$ is a multiple of
\[
k_2 := \frac{k}{\mathfrak{g}} = \frac{k}{\gcd(k,i,p)}
\]
and $k_1k_2$ divides the product~\eqref{eq:prod_bin}. Finally, since a common
divisor of $k$ and $i$ cannot divide $2k-i - 1$, $\gcd(k,i,p)$ and $\gcd(k,2k -
i - 1)$ are mutually prime; therefore, $k$ is a factor of $k_1k_2$.
\end{proof}

A natural question would be to ask whether the cumulants of the Wigner
delay time for the symmetry classes $\beta=1$ and $\beta=4$ are integers 
too. In Table~\ref{pie} we list the leading order coefficients of the
first eight cumulants, which we computed numerically, for $\beta \in
\left \{1,2,4\right\}$.  We immediately see that when $\beta=4$ this
is not the case; when $\beta=1$ it remains an open
question.

\begin{table}[h]
\caption{The limiting cumulants 
of the Wigner delay time for each $\beta \in \{1,2,4\}$}
\centering
    \begin{tabular}{  |c | c |c | c| }
    \hline\hline
    l & $K^{(1)}_{l}(0)$ & $K^{(2)}_{l}(0)$ & $K^{(4)}_{l}(0)$ \\ \hline
    $1$ & $1$ & $1$ & $1$ \\ 
    $2$ & $4$ & $2$ & $1$\\ 
    $3$ & $96$ & $4$ & $6$  \\ 
    $4$ & $5088$ & $636$ & $159/2$ \\ 
    $5$ & $437760$ & $27360$ & $1710$ \\ 
    $6$ & $53038080$ & $1657440$ & $51795$ \\ 
    $7$ & $8353013760$ & $130515840$ & $2039310$ \\ 
    $8$ & $1625430159360$ & $12698673120$ & $396833535/4$ \\
    \hline\hline
    \end{tabular}
\label{pie}
\end{table}

\section{Conclusions and Open Problems}
\label{se:conc}

In this paper we study several properties of the cumulants and of
their generating functions of the conductance, shot noise and Wigner
delay time in ballistic quantum dots for all the symmetry classes
$\beta \in \left \{1,2,4 \right \}$.  We focus on the behaviour for
finite numbers of quantum channels $n$ and at leading order as $n \to 
\infty$.  The number of channels can be interpreted as the
semiclassical parameter.

The main results are the limiting formulae as $n\to \infty$ of the
cumulants given in Theorems~\ref{maintheorem}
and~\ref{th:lim_cum_wdt}, as well as the differential
equations~\eqref{diffeqshot},~\eqref{diffeqcond} and~\eqref{wigode},
which are satisfied by the cumulant generating functions at finite
$n$.  Theorem~\ref{maintheorem} provides the weak localization
corrections of the mixed cumulants of the conductance and shot noise
for the symmetry classes $\beta=1,4$. Theorem~\ref{th:lim_cum_wdt}
concerns the leading order asymptotics of the cumulants of the Wigner
delay time when $\beta=2$.  The differential
equations~\eqref{diffeqshot}, \eqref{diffeqcond} and~\eqref{wigode}
apply to all symmetry classes $\beta \in \left \{1,2,4\right\}$.  When
$\beta=2$ we show that the cumulant generating function of the Wigner
delay time can be expressed in terms of the Painlev\'e
III${}^{\prime}$ transcendent and discover an interesting identity
between the probability distributions of the conductance and shot
noise.  Finally, we derive finite $n$ as well as asymptotic recurrence
relations that are very effective to compute cumulants of arbitrary
order.

The results for the conductance and shot noise apply not only to the
Wigner-Dyson symmetry classes, but also to those introduced by
Zirnbauer~\cite{Zir96} and Altland and Zirnbauer~\cite{AZ96,AZ97},
which characterize superconducting quantum dots (also known as Andreev
quantum dots).  However, the formulae for the Wigner delay time are
limited to the Wigner-Dyson symmetries.

A number of conjectures proposed by Khoruzhenko \textit{et
  al.}~\cite{KSS09} are proved as particular cases of
Theorem~\ref{maintheorem}.  They concern the weak localization
corrections of all the cumulants of the conductance (see
Eqs.~\eqref{goddconj} and~\eqref{gevenconj}) as well those of the even
cumulants of the shot noise (see Eq.~\eqref{pevenconj}).  In the same
article they conjectured the leading order~\eqref{poddconj} of the odd
cumulants of the shot noise.  Unfortunately, the latter are of
subleading order to our formulae; their proof remains an open problem.

Our study is far from being exhaustive and several interesting
questions remain unsolved.  Our approach is based on the integrable
theory of $\tau$-functions that are deformations of certain matrix
integrals.  These techniques were mainly developed by Adler and van
Moerbeke~\cite{AvM95,AvM01a,AvM01b,AvM02} and Adler \textit{et
  al.}~\cite{ASvM95,ASvM98,ASvM02}.  While these ideas appear to be
very effective to study the leading order asymptotics of the
cumulants, going beyond the leading order seems a very difficult task.
We discuss this aspect in Sec.~\ref{higherorder}.

Another challenging project would be to compute the leading order
asymptotics of the cumulants of the Wigner delay time when 
$\beta=1,4$.  Such computations would involve a double-scaling limit
that appears to be intractable using the techniques in this paper.
Interestingly, we discover that at leading order the cumulants of the
Wigner delay time are integers when $\beta=2$, but are not
when $\beta=4$; when $\beta=1$ it remains an open problem.  


There is some evidence that the limiting formulae in Theorem
\ref{maintheorem} hold for any real $\beta>0$ (see Remark~\ref{ntc}),
\textit{i.e.} for the $\beta$-ensembles, which are of significant
interest in many applications of modern RMT. Generalizing the techniques 
applied in this paper to arbitrary $\beta>0$ poses some fundamental problems, namely
it is not known whether there exist analogues of Eqs. (\ref{kpequation}) and (\ref{pfaff}) for other values of $\beta$. 

We conclude by pointing out that Dahlhaus \textit{et al.}~\cite{DBB10}
have suggested that the Andreev quantum dots of symmetry types DIII and CI 
-- in the Cartan classification --- could be realized experimentally
using graphene.  This opens the possibility that some of our formulae 
could be verified experimentally.



\section*{Acknowledgements}

We would like to express our gratitude to Dmitry Savin for
interesting discussions on the topics in this paper.

\appendix

\section{Proof of the Differential Equation~\eqref{diffeqcond}}
\label{virconapp}

The derivations of the ODEs in
Theorems~\ref{pde:mix_csn},~\ref{ode:con} and~\ref{th:wdt_ode} are
based on the same approach. The projections at the origin of some the
partial derivatives of the $\tau$-function in the KP and Pfaff-KP
equations~\eqref{kpequation} and~\eqref{pkpequation} are expressed in
terms of the derivatives of the moments generating function using the
Virasoro constraints.  The calculations are involved, but standard in
the theory of integrable systems.  Here we discuss the derivation of
the ODE~\eqref{diffeqcond}, which is satisfied by the cumulant
generating function of the conductance.  In this case we need to write
the projections~\eqref{missingpds} in terms of
$\sigma_n^{(\beta)}(z)$ and its derivatives.

For ease of notation, let us define
\begin{equation*}
  o_{1} := (n\beta+2-\beta),\quad  o_{2} := (n\beta+3-3\beta/2), \quad
c: =\alpha n+n\bigl((n-1)\beta/2+1\bigr).
\end{equation*}
Recall that
\begin{equation*}
  \sigma_n^{(\beta)}(z) = \log \mathcal M_n^{(\beta)}(z) = \log
  \tau_n(\mathbf{t};z)\bigr \rvert_{\mathbf{t}=\mathbf{0}}
\end{equation*}
  and
\begin{equation*}
 \frac{\partial^{k} \log \tau_n(\mathbf{t};z)}{\partial t_{1}^{k}}
\bigg|_{\mathbf{t}=0} = (-1)^{k}\frac{d^{k}\sigma_n^{(\beta)}(z)}{dz^{k}}.
\end{equation*}

Setting $w=0$ in (\ref{shotvira1}) and (\ref{shotvira2}), we arrive at
the following Virasoro constraints:
\begin{subequations}
\begin{align}
\label{logvc1}
\frac{\mathcal{V}_{-1}\tau_n(\mathbf{t};z)}{\tau_{n}(\mathbf{t};z)} &
=\left(\sum_{i=1}^{\infty}it_{i}\left(\frac{\partial}{\partial
      t_{i+1}} -\frac{\partial}{\partial t_{i}}\right)\right .\\
& \quad \left.  +(\alpha+z+\delta/2+o_{1})\frac{\partial}{\partial t_{1}}
-z\frac{\partial}{\partial t_{2}}\right)\log \tau_n(\mathbf{t};z)-c=0 \notag\\ 
\intertext{and}
\label{logvc2}
\frac{\mathcal{V}_{0}\tau_n(\mathbf{t};z)}%
  {\tau_{n}(\mathbf{t};z)} &= \left(\sum_{i=1}^{\infty}it_{i}
   \left(\frac{\partial}{\partial t_{i+2}}-\frac{\partial}{\partial
       t_{i+1}}  \right)+\frac{\beta}{2}\frac{\partial^{2}}%
  {\partial t_{1}^{2}}-(o_{1}+\alpha)\frac{\partial}{\partial t_{1}}\right.\\
& \quad \left.+(\alpha+z+\delta/2+o_{2})\frac{\partial}{\partial t_{2}}
  -z\frac{\partial}{\partial t_{3}}\right)\log
\tau_n(\mathbf{t};z)\notag \\
 & \quad +\frac{\beta}{2}\left(\frac{\partial 
 \log \tau_n(\mathbf{t};z)}{\partial t_{1}}\right)^{2}=0. \notag 
\end{align}
\end{subequations}

Setting $\mathbf{t}=\mathbf{0}$ in \eqref{logvc1} and \eqref{logvc2},
respectively, leads to the relations
\begin{align*}
X_{1} := \frac{\partial \log \tau_n(\mathbf{t};z)}%
{\partial t_{2}}\bigg|_{\mathbf{t}=\mathbf{0}} & = 
- \frac{1}{z}\bigl(c+(\alpha+z+\delta/2+o_{1})\sigma_n^{(\beta)\prime}\bigr),\\
X_{2} := \frac{\partial \log \tau_n(\mathbf{t};z)}{\partial t_{3}}
\bigg|_{\mathbf{t}=\mathbf{0}} &=\frac{1}{z}\left(\frac{\beta}{2}
\sigma_n^{(\beta)\prime \prime }+\frac{\beta}{2}\bigl(\sigma_n^{(\beta)
  \prime}\bigr)^{2} \right.\\
&\quad +(o_{1}+\alpha)\sigma_n^{(\beta) \prime}
+(\alpha+z+\delta/2 +o_{2})X_{1}\Bigr).
\end{align*}

Differentiating both \eqref{logvc1} and \eqref{logvc2} with respect to
$t_{1}$ and setting $\mathbf{t}=\mathbf{0}$ yields two further
relations
\begin{align*}
X_{3} := \frac{\partial^{2}\log \tau_n(\mathbf{t};z)}{\partial t_{1}
 \partial t_{2}}\bigg|_{\mathbf{t}=\mathbf{0}} &= \frac{1}{z}
\left(\sigma_n^{(\beta)\prime}+X_{1}+(\alpha+z+d/2+o_{1})\sigma_n^{(\beta)\prime 
  \prime}\right),\\
X_{4} := \frac{\partial^{2}\log \tau_n(\mathbf{t};z)}
{\partial t_{1} \partial t_{3}}\bigg|_{\mathbf{t}=\mathbf{0}} &= \frac{1}{z}
\Bigl(X_{2}-X_{1}-(\beta/2)\bigl(\sigma_n^{(\beta) 
\prime \prime \prime}+2\sigma_n^{(\beta)\prime}
\sigma_n^{(\beta)\prime \prime}\bigr)\\
&\quad -(o_{1}+\alpha)\sigma_n^{(\beta)\prime
    \prime}+(\alpha+z+\delta/2  +o_{2})X_{3}\Bigr). 
\end{align*}
Taking the partial derivatives of the constraint~\eqref{logvc1} 
with respect to $t_{2}$ and setting $\mathbf{t}=\mathbf{0}$ gives 
\begin{equation*}
X_{5} := \frac{\partial^{2}\log \tau_n(\mathbf{t};z)}{\partial t_{2}^{2}}
\bigg|_{\mathbf{t}=\mathbf{0}} = \frac{1}{z}\bigl(2X_{2}-2X_{1}
+(\alpha + \delta/2 +z+o_{1})X_{3}\bigr).
\end{equation*}
Therefore, each $X_{i}$, $i=1,\dotsc,5$, can be expressed solely in
terms of known parameters and derivatives of $\sigma^{(\beta)}_{n}(z)$.
Finally, evaluating the Pfaff-KP and KP equations~\eqref{pkpequation}
and~\eqref{kpequation} at $\mathbf{t}=\mathbf{0}$ and inserting
$X_{4}$ and $X_{5}$ yields~\eqref{diffeqcond}.

\section{The Constants $b^{(\beta)}_{n}$ and $d^{(\beta)}_{n}$}
\subsection{Conductance and Shot Noise}
\label{bnapp}
The constant $b^{(\beta)}_{n}$, as introduced in Eq.~\eqref{bnbeta},
is defined in terms of Selberg's integral, whose explicit expression
is given in Eq.~\eqref{eq:selberg}.

The  ratio of Selberg's integrals in \eqref{bnbeta} 
involves a large number of cancellations of Gamma
functions. Straightforwards algebra gives 
\begin{subequations}
\label{bnforms}
\begin{align}
\label{bn1form}
b^{(1)}_{n} & := \frac{n(n-1)}{(n+1)(n+2)}\frac{c^{(1)}_{n-2}c^{(1)}_{n+2}}%
{\bigl(c^{(1)}_{n}\bigr)^{2}} \\
&\;= \frac{n(n-1)(2\alpha+n)(2\alpha+n+1)(\delta
  + n)(\delta+n+1)}{16(\delta/2+\alpha+n)(\delta/2+\alpha+n+1)^{2}
(\delta/2+\alpha+n+2)}\notag \notag \\
&\quad \times \frac{(\delta+2\alpha+n+1)(\delta+2\alpha+n+2)}%
{(\delta+2\alpha+2n-1)(\delta+2\alpha+2n+1)^{2}(\delta+2\alpha+2n+3)},\notag
\intertext{while for $\beta=4$}
\label{bn4form}
b^{(4)}_{n} & :=
\frac{n}{n+1}\frac{c^{(4)}_{n-1}c^{(4)}_{n+1}}{\bigl(c^{(4)}_{n}\bigr)^{2}} \\
&\;  = \frac{2n(2n+1)(\alpha+2n)(\alpha+2n-1)(\delta/2+2n)(\delta/2+2n-1)}%
{(\delta/2+\alpha+4n)(\delta/2+\alpha+4n-2)^{2}(\delta/2+\alpha+4n-4)}\notag\\
&\quad\times \frac{(\delta/2+\alpha+2n-1)(\delta/2+\alpha+2n-2)}%
{(\delta/2+\alpha+4n+1)(\delta/2+\alpha+4n-1)^{2}(\delta/2+\alpha+4n-3)}.\notag 
\end{align}
\end{subequations}

\subsection{Wigner Delay Time}
\label{ap:dn}
The parameter $d^{(\beta)}_{n}$, see Eq.~\eqref{dnbeta},
has the same structure as $b_n^{(\beta)}$, with the difference that
Selberg's integral is replaced by the normalization constant
$m^{(\beta)}_{n}$, which appears in the definition of the partition
function~\eqref{wignermgf}.  Similar calculations to those leading
to Eqs.~\eqref{bn1form}  and~\eqref{bn4form} give
\begin{align*}
d^{(1)}_{n} & := \frac{n(n-1)}{(n+1)(n+2)}
\frac{m^{(1)}_{n-2}m^{(1)}_{n+2}}{\bigl(m^{(1)}_{n}\bigr)^{2}}\\
&= \frac{n(n-1)(2b-2-n)(2b-1-n)}%
{(b-n)(2b-2n+1)(b-n-2)(2b-2n-3)(b-1-n)^{2}(2b-2n-1)^{2}}
\intertext{and}
d^{(4)}_{n} & :=
\frac{n}{n+1}\frac{m^{(4)}_{n-1}m^{(4)}_{n+1}}{\bigl(m^{(4)}_{n}\bigr)^{2}}\\
  &=  \frac{2n(2n+1)(b+2-2n)(b+1-2n)}%
{(b+3-4n)(b+1-4n)^{2}(b+2-4n)^{2}(b-1-4n)(b-4n)(b-4n+4)}.
\end{align*}
\subsection{Dualities}
One easily sees that the constants $b_n^{(\beta)}$ and $d_n^{(\beta)}$
display the following dualities between the symmetry classes
$\beta=4$ and $\beta=1$:
\begin{subequations}
\begin{align}
\label{jacobidual}
b^{(4)}_{n} &= b^{(1)}_{n}\Bigr\rvert_{\delta \to -\delta/2, 
\alpha \to -\alpha/2, n \to -2n},\\
16d^{(4)}_{n} &= d^{(1)}_{n}\Bigr\rvert_{b \to -b/2,n \to -2n}.
\end{align}
\end{subequations}
The duality (\ref{jacobidual}), and other similar relations, were also
mentioned by Adler and van Moerbeke~\cite{AvM01a}.

\section{Proof of  Eq.~\eqref{lrsesol}}
\label{ap:p_th}
We now prove that \eqref{lrsesol} is the unique solution of the
initial value problem \eqref{lrse} with initial conditions
\eqref{conic1}.

First observe that the right-hand side of Eq.~\eqref{lrsesol} is
expressible in terms of Jacobi polynomials using the identity
\begin{equation}
\label{jident}
\sum_{j=0}^{k}\binom{2j+2l}{j+l}
\binom{k}{j}(-2)^{-j} = \frac{(2l)!k!}{l!(k+l)!}P^{(l,-k-1/2)}_{k}(-3).
\end{equation}
Recall that the Jacobi Polynomials can be defined by
\begin{equation*}
P^{(a,b)}_{n}(x) := \sum_{j=0}^{n}\binom{n+a}{j}\binom{n+b}{j+b}
\left(\frac{x-1}{2}\right)^{n-j}\left(\frac{x+1}{2}\right)^{j}.
\end{equation*}
The identity~\eqref{jident} is a consequence of this
formula.

Inserting the right-hand side of \eqref{jident} into \eqref{lrsesol}
and substituting $\tilde \kappa_{2l,k}$ into~\eqref{lrse} gives
\begin{multline}
\label{proveme1}
(4l+3k+2)(k+1)P^{(l,-k-3/2)}_{k+1}(-3)+(2l+1)(k+l+1)
P^{(l+1,-k+1/2)}_{k-1}(-3)\\
+(2l+k)(2l+1)P^{(l+1,-k-1/2)}_{k}(-3)-3(k+l)(k+l+1)P^{(l,-k+1/2)}_{k-1}(-3)=0.
\end{multline}
This identity can be proved using well known properties of the Jacobi
polynomials. Firstly, we identify the parameters of each Jacobi
polynomial in \eqref{proveme1} using the identities (see,
\textit{e.g.},~\cite{AS72}, Eqs.~22.7.18 and 22.7.19)
\begin{align*}
(2n+a+b)P^{(a,b-1)}_{n}(x)&=(n+a+b)P^{(a,b)}_{n}(x)+(n+a)P^{(a,b)}_{n-1}(x),\\
(2n+a+b)P^{(a-1,b)}_{n}(x) &= (n+a+b)P^{(a,b)}_{n}(x)-(n+b)P^{(a,b)}_{n-1}(x).
\end{align*}
Then, we repeatedly apply the three-term recurrence relation (see,
\textit{e.g.},~\cite{AS72}, Eq.~22.7.1)
\begin{equation*}
\begin{split}
xP^{(a,b)}_{n}(x)&=\frac{2(n+1)(a+b+n+1)(2n+a+b)P^{(a,b)}_{n+1}(x)}%
{(a+b+2n)(a+b+2n+1)(a+b+2n+2)}\\
& \quad +\frac{(2n+a+b+1)(b^{2}-a^{2})P^{(a,b)}_{n}(x)
+2(a+n)(b+n)(a+b+2n+2)P^{a,b}_{n-1}(x)}{(a+b+2n)(a+b+2n+1)(a+b+2n+2)}. 
\end{split}
\end{equation*}

We used a similar technique to prove equalities such
as~\eqref{proveme1} in a previous publication~\cite{MS12} (see,
\textit{e.g.}, Proposition 3.10).

\section{The Solution of the Recurrence Relation~\eqref{rec1}}
\label{se:solrec}

The recurrence relations~\eqref{rec1} and~\eqref{limwigrecintro} can
be solved using the method of generating functions.  We now give a
flavour of the main ideas involved by discussing in some detail the
example of Eq.~\eqref{rec1}.

The change of variables~\eqref{eq:r_cum} turns Eq.~\eqref{rec1}
into~\eqref{rec1_r}. Let us first consider the even cumulants. By
replacing $l$ with $2l-1$ in \eqref{rec1_r}, we obtain
\begin{equation}
\label{wigevencums}
-2l\tilde{\kappa}_{2l}+(2l-1)(2l-2)(8l-9)\tilde{\kappa}_{2l-2}
=3(2l-1)(2l-2)(2l-3)(2l-4)(2l-5)\tilde{\kappa}_{2l-4},
\end{equation}
which we need to solve subject to the initial conditions
$\tilde{\kappa}_{4}=3/4$ and $\tilde{\kappa}_{6}=15/2$. Introducing
the sequence $p_{l}=\tilde{\kappa}_{2l}/(2l-1)!$, $l=2,3,\dotsc,$ the
recurrence \eqref{wigevencums} becomes
\begin{equation}
\label{precapp}
-2lp_{l}+(8l-9)p_{l-1}=3(2l-5)p_{l-2},
\end{equation}
with initial conditions $p_{2}=1/8$ and
$p_{3}=1/16$. 

The difference equation~\eqref{precapp} is transformed into a
differential equation for the generating function $\tilde{F}_{\rm
  e}(y)$, defined in \eqref{eq:res_gf_even}, by multiplying both side
by $y^l$ and summing over $l$ from $4$ to $\infty$. For example, we
have
\begin{equation*}
\sum_{l=4}^{\infty}(-2lp_{l})y^{l}=-2y\frac{d}{dy}
\sum_{l=4}^{\infty}p_{l}y^{l}=-2y\frac{d}{dy}
\left(\tilde{F}_{\rm e}(y)-y^{2}/8-y^{3}/16\right).
\end{equation*}
Repeating this procedure for each term in (\ref{precapp}) and
simplifying leads to
\begin{equation}
\label{precapp2}
y(3y-1)\left(y-2\tilde{F}_{\rm e}+4(y-1)\frac{d\tilde{F}_{\rm e}}{dy}\right)=0,
\end{equation}
with initial condition $\tilde F_{\rm e}(0)=0$.

For any $y$ in the punctured disk $\{y \in \mathbb{C} \mid 0 < |y| <
1/3\}$, we divide both sides of~\eqref{precapp2} by $y(3y-1)$ and 
solve the resulting first order linear equation by the method of integrating 
factors. The solution is 
\begin{equation*}
\tilde{F}_{e}(y) = 1-\frac{y}{2}- \sqrt{1-y}.
\end{equation*}

The odd cumulants can be found in a similar fashion.  The variables
$p_{l}=\tilde{\kappa}_{2l+1}/(2l)!$ satisfy the recurrence
\begin{equation}
\label{precapp3}
-(2l+1)p_{l}+(8l-5)p_{l-1}=6(l-1)p_{l-2},
\end{equation}
with initial conditions $p_{1}=1$ and $p_{2}=1$. Then, one can show
that the generating function $\tilde{F}_{\rm o}(y)$, defined in
Eq.~\eqref{eq:res_gf_odd} satisfies the linear ODE
\begin{equation*}
  2y(1-3y)(1-y)\frac{d\tilde{F}_{\rm o}(y)}{dy} + (6y^2 -3y +
  1)\tilde{F}_{\rm o}(y) -  3y(1-2y)=0, 
\end{equation*}
with initial condition $\tilde{F}_{\rm o}(0)=0$. This equation can be
solved by the method of integrating factors. The solution is
\begin{equation*}
\tilde{F}_{\rm o}(y)=\frac{y}{1-y}.
\end{equation*}

Alternatively, one easily verifies by direct substitution that
$p_{l}=1$ for all $l>0$ solve Eq.~\eqref{precapp3}.

\bibliographystyle{amsplain}
\bibliography{chaotic_cav_bib}

\end{document}